\documentclass[aps,twocolumn,longbibliography,superscriptaddress]{revtex4-1}%
\usepackage{amsfonts}
\usepackage{amssymb}
\usepackage{amsmath,mathtools}
\usepackage{graphicx}
\usepackage{dsfont}
\usepackage{hyperref}
\hypersetup{colorlinks=true,citecolor=blue,linkcolor=blue,filecolor=blue,urlcolor=blue,breaklinks=true}
\usepackage{times}
\providecommand{\U}[1]{\protect\rule{.1in}{.1in}}
\newtheorem{theorem}{Theorem}

\newtheorem{definition}{Definition}

\newtheorem{proposition}{Proposition}
\newtheorem{remark}[theorem]{Remark}

\newenvironment{proof}[1][Proof]{\noindent\textbf{#1.} }{\ \rule{0.5em}{0.5em}}

\newcommand{\nc}{\newcommand}
\nc{\EPPT}{{E_{\operatorname{PPT}}}}
\nc{\EPPTone}{{E_{\operatorname{PPT}}^{(1)}}}
\nc{\EK}{{E_{\kappa}}}
\nc{\1}{{\mathds{1}}}
\nc{\lbar}[1]{\overline{#1}}
\nc{\bra}[1]{\langle#1|}
\nc{\ket}[1]{|#1\rangle}
\nc{\ketbra}[2]{|#1\rangle\!\langle#2|}
\nc{\braket}[2]{\langle#1|#2\rangle}

\nc{\proj}[1]{| #1\rangle\!\langle #1 |}
\nc{\avg}[1]{\langle#1\rangle}
\nc{\Rank}{\operatorname{Rank}}
\nc{\smfrac}[2]{\mbox{$\frac{#1}{#2}$}}
\nc{\tr}{\operatorname{Tr}}
\nc{\ox}{\otimes}
\nc{\dg}{\dagger}
\nc{\dn}{\downarrow}
\nc{\cA}{{\cal A}}
\nc{\cB}{{\cal B}}
\nc{\cC}{{\cal C}}

\nc{\PPT}{{\operatorname{PPT}}}

\begin{document}
\preprint{ }
\title{Cost of quantum entanglement simplified}
\author{Xin Wang}
\affiliation{Institute for Quantum Computing, Baidu Research, Beijing 100193, China}
\affiliation{Joint Center for Quantum Information and Computer Science, University of
Maryland, College Park, Maryland 20742, USA}
\email{wangxin73@baidu.com}
\author{Mark M. Wilde}
\affiliation{Hearne Institute for Theoretical Physics, Department of Physics and Astronomy, and 
Center for Computation and Technology, Louisiana\ State University, Baton
Rouge, Louisiana 70803, USA}
\email{mwilde@lsu.edu}

\begin{abstract}
Quantum entanglement is a key physical resource in quantum information
processing that allows for performing basic quantum tasks such as teleportation
and quantum key distribution, which are impossible in the classical world.
Ever since the rise of quantum information theory, it has been an open problem to quantify entanglement in an information-theoretically   meaningful way. In particular, every previously defined entanglement measure bearing a precise information-theoretic meaning is not known to be efficiently computable, or if it is efficiently
computable, then it is not known to have a precise information-theoretic meaning. In this paper, we
meet this challenge by introducing an entanglement measure that has a precise
information-theoretic meaning as the exact cost required to prepare an entangled state when
two distant parties are allowed to perform quantum operations that completely
preserve the positivity of the partial transpose. 
Additionally, this
entanglement measure is efficiently computable by means of a semi-definite
program, and it bears a number of useful properties such as additivity and faithfulness. Our results bring key insights into the fundamental
entanglement structure of arbitrary quantum states, and they can
be used directly to assess and quantify the entanglement produced in
quantum-physical experiments.
\end{abstract}
\date{\today}
\startpage{1}
\endpage{10}
\maketitle

\textit{Introduction.}---Quantum entanglement is a fundamental property of quantum states that has no
classical analog. As famously remarked by Schr\"{o}dinger \cite{S35}, it is
\textquotedblleft the characteristic trait of quantum mechanics, the one that
enforces its entire departure from classical lines of
thought.\textquotedblright\ Einstein, Podolsky, and Rosen were confounded by
entanglement \cite{epr1935}, and based on this, proposed a theory alternative
to quantum mechanics, which was later ruled out by a theoretical proposal of
Bell \cite{bell1964} and experimental confirmations of Bell's test
\cite{PhysRevLett.47.460,Hensen2015,PhysRevLett.115.250401,PhysRevLett.115.250402}%
.

The aforementioned early work on understanding entanglement ended up being
foundational for the modern field of quantum information science
\cite{book2000mikeandike,W17book}, whose goal is to harness the strange
properties of quantum states for information processing tasks that are not
possible in the classical world. Due to seminal work by Bennett \textit{et
al}., we now understand quantum entanglement to be the enabling fuel for a
variety of quantum protocols such as teleportation \cite{Bennett1993}, dense
coding \cite{PhysRevLett.69.2881}, and quantum key distribution
\cite{bb84,Ekert:1991:661}.

In the fundamental protocols mentioned above, it is required for the entangled
states being consumed to be in a pure form, known as maximally entangled
states. However, in experimental practice, quantum states do
not come in this pure variety, but instead are produced as mixtures of pure
states. As such, a key goal of the resource theory of entanglement
\cite{Bennett1996c} is to understand how well mixed quantum states can be
converted to pure maximally entangled states and vice versa, by means of
\textquotedblleft free\textquotedblright\ physical operations that do not
increase entanglement.
Motivated by the \textquotedblleft distant laboratories
paradigm,\textquotedblright\ in which the two parties holding shares of a
quantum state are spatially separated, one set of physical operations that is reasonable to 
allow for free consists of those that can be implemented by local operations and
classical communication (LOCC). The characterization of entanglement as a resource in practical settings is also rooted in this distant laboratories paradigm.

There are two primary operational ways for quantifying entanglement in a
two-party quantum state $\rho_{AB}$:\ the first is known as distillable
entanglement \cite{Bennett1996c}\ and the second is known as entanglement cost
\cite{Bennett1996c,Hayden2001}. In the first approach, one is interested to
know the largest rate at which maximally entangled states can be distilled by
means of LOCC\ from the state $\rho_{AB}$. In the second, one is interested to
know the smallest rate at which maximally entangled states are required to
prepare the state $\rho_{AB}$ by means of LOCC. There are a number of
technical variations of each task that have been considered
\cite{Bennett1996c,TH00,Hayden2001,BD11},\ involving one or multiple copies of
the state $\rho_{AB}$, or for the task to be accomplished exactly or with some
error tolerance. So far, beyond the case of pure states \cite{BBPS96}, it has
been a great challenge since the publication of the seminal work in
\cite{Bennett1996c} to characterize the distillability and cost of quantum
entanglement. Much of the progress during the past two decades has to do with
finding alternative entanglement measures that bound entanglement
distillability or cost, while possessing properties that are generally agreed
upon to be reasonable
\cite{VP98,R99,R01,Vidal2002,CW04,Plenio2005b,Plenio2007,Horodecki2009a,WD16pra,Wang2016d}. Even many of the measures that have been defined are known to be difficult to compute \cite{Huang2014}.

Due to the aforementioned challenges associated with mixed-state entanglement
and the set LOCC in general \cite{CLMOW14}, researchers have looked in other
directions in order to understand the nature of entanglement. One approach was
pioneered in \cite{R01}, with the introduction of another set of free
operations that \textquotedblleft completely preserve the positivity of the
partial transpose\textquotedblright\ \cite{CVGG17} (we explain the precise meaning of this
term later). This set (abbreviated by C-PPT-P) has been considered in prior
work \cite{PhysRevLett.87.257902,Audenaert2003,PhysRevLett.89.240403,Plenio2005b,Wang2016d} on entanglement theory for at least two reasons:
\begin{enumerate}
\item The mathematical
structure of LOCC\ is difficult to work with and so enlarging the set to a
more mathematically tractable set allows for providing bounds on what one
could accomplish with LOCC. That is, such free operations  provide accessible estimates of the capabilities of LOCC in entanglement manipulation. 

\item The free states that correspond to this
enlarged set, known as positive partial transpose (PPT)\ states, in any case
do not have any useful entanglement on their own (in the sense that it is 
impossible to distill maximally entangled states from them at a non-trivial rate via LOCC).
\end{enumerate}
\noindent As observed in~\cite{R01}, an
advantage of the set C-PPT-P over LOCC\ is that performing optimizations over
it allows for incorporating the tools of semi-definite programming.

One key problem that has remained open for many years now is to characterize
the exact entanglement cost of a quantum state $\rho_{AB}$\ when
C-PPT-P\ operations are allowed for free, which is equal to the minimum rate
at which entanglement is required to prepare many perfect and identical copies
of $\rho_{AB}$ by means of these free operations. The optimal rate is known as
the PPT\ exact entanglement cost. The problem was formalized in
\cite{Audenaert2003}, where some bounds on this quantity were given and some
partial solutions were presented.\ The problem was considered further in
\cite{Matthews2008}, which however focused mainly on transformations of pure
entangled states. 

In this paper, we determine the PPT\ exact entanglement cost of an arbitrary
two-party quantum state, thus closing a longstanding investigation in
entanglement theory. We find that the solution is given by a new entanglement
measure, which we call $\kappa$-entanglement. The $\kappa$-entanglement can be
calculated by means of a semi-definite program \cite{BV04}, implying that it can be
efficiently calculated in time polynomial in the dimension of the state on
which it is being evaluated \cite{Khachiyan1980,AHK05,AK07,AHK12, LSW15}. These two properties single out the $\kappa
$-entanglement as the first entanglement measure that has a concrete
information-theoretic meaning while being efficiently calculable. The $\kappa
$-entanglement also bears a number of desirable properties, including
additivity, normalization, and faithfulness, which we expand upon later. It is
neither convex nor monogamous \cite{T04}, which calls into question whether it
is truly necessary for an entanglement measure to satisfy either of these properties.

Our results on $\kappa$-entanglement of quantum states bring new
insights regarding the structure of quantum entanglement as a physical
resource. For one, they demonstrate that entanglement can be quantified in a
precise and physically relevant operational scenario. Furthermore, they call
into question whether properties such as monogamy or convexity are really
required for entanglement measures, in light of the fact that $\kappa
$-entanglement does not have these properties while at the same time having
the aforementioned operational meaning.

We now begin the more technical part of our paper by giving some background, defining the PPT exact entanglement cost and $\kappa$-entanglement of a quantum state, and justifying how these quantities are equal. We note here that all mathematical proofs of the various statements and
properties summarized in this paper are given in the Supplementary Material \cite{ECappendix}, which also includes the following references not mentioned in the main text \cite{N99,H06book,YC18,Horodecki99,Watrous2011b,FAR11,Ishizaka2004a,CKW00,KWin04,D67,RS78,GS19,WW18,FB05}. 

Let us first recall some basic elements of quantum information. A two-party
or bipartite quantum state $\rho_{AB}$ is a unit trace, positive semi-definite
operator acting on a tensor-product Hilbert space $\mathcal{H}_{A}%
\otimes\mathcal{H}_{B}$. We say that Alice possesses system $A$ and Bob system
$B$, and we imagine that Alice and Bob are located in distant laboratories.
Such a state is separable \cite{W89} if there exists a probability
distribution $p_{X}$ and sets of states $\left\{  \sigma_{A}^{x}\right\}
_{x}$ and $\{\tau_{B}^{x}\}_{x}$ such that%
\begin{equation}
\rho_{AB}=\sum_{x}p_{X}(x)\sigma_{A}^{x}\otimes\tau_{B}^{x}.
\end{equation}
If $\rho_{AB}$ cannot be written in the above way, then it is entangled~\cite{W89}.

It is a difficult (NP-hard) computational problem to decide whether an arbitrary quantum
state is separable or entangled \cite{G03,G10}. As such, researchers have
sought out simpler, \textquotedblleft one-way\textquotedblright\ criteria to
classify entanglement of quantum states. Possibly the simplest such criterion
is the positive partial transpose criterion \cite{Per96,HHH96}. To define
this, recall that the partial transpose, with respect to a given orthonormal
basis $\{|i\rangle_{B}\}_{i}$,\ is defined as the following linear map:%
\begin{equation}
T_{B}(X_{AB}):=\sum_{i,j}\left(  I_{A}\otimes|i\rangle\langle j|_{B}\right)
X_{AB}\left(  I_{A}\otimes|i\rangle\langle j|_{B}\right)  ,
\end{equation}
which we also write as $X_{AB}^{T_{B}} \equiv T_{B}(X_{AB})$.
An operator $X_{AB}$ has positive partial transpose (PPT) if $T_{B}(X_{AB})$
is positive semi-definite. By inspecting definitions, we conclude that if a
bipartite state is separable, then it is a PPT state. By contrapositive, we
conclude that a bipartite state is entangled if it has a negative partial transpose. The PPT criterion is one-way in the sense that there exist entangled PPT states \cite{HHH98}. 

A quantum channel is a completely positive trace-preserving map, and a bipartite quantum channel $\mathcal{N}_{AB\rightarrow A^{\prime}B^{\prime}}$
accepts input systems $A$ and $B$ and outputs $A^{\prime}$ and $B^{\prime}$,
where one party Alice possesses $A$ and $A^{\prime}$ and another party Bob possesses $B$ and
$B^{\prime}$. A bipartite quantum channel $\mathcal{N}_{AB\rightarrow
A^{\prime}B^{\prime}}$ is completely positive-partial-transpose preserving
\cite{R01}, abbreviated as C-PPT-P,\ if the map $T_{B^{\prime}}%
\circ\mathcal{N}_{AB\rightarrow A^{\prime}B^{\prime}}\circ T_{B}$ is
completely positive.

In the resource theory of NPT\ (non-positive partial transpose) entanglement,
the free operations allowed are C-PPT-P\ bipartite channels and the free states are PPT states, and one of the
main goals is to determine if one bipartite state can be converted to another
either exactly or approximately by means of the free operations. The
particular task of interest to us here is the PPT exact entanglement cost. 

\textit{One-shot exact entanglement cost.}---We
begin by defining the one-shot PPT\ exact entanglement cost of a bipartite
state $\rho_{AB}$\ as the logarithm of the minimum Schmidt rank of a maximally
entangled state that is required to prepare $\rho_{AB}$ by means of a C-PPT-P
channel:%
\begin{equation*}
E_{\operatorname{PPT}}^{(1)}(\rho_{AB}):=\log_{2}\inf_{d\in\mathbb{N},\Lambda\in
\operatorname{PPT}}\left\{  d: \rho_{AB} = \Lambda_{\hat{A}\hat{B}\rightarrow AB}%
(\Phi_{\hat{A}\hat{B}}^{d})\right\}  ,
\end{equation*}
where $\Lambda\in
\operatorname{PPT}$ is a shorthand for $\Lambda$ being a C-PPT-P bipartite channel and the maximally entangled state $\Phi_{\hat{A}\hat{B}}^{d}$ is defined as%
\begin{equation}
\Phi_{\hat{A}\hat{B}}^{d}:=\frac{1}{d}\sum_{i,j}|i\rangle\langle j|_{\hat{A}%
}\otimes|i\rangle\langle j|_{\hat{B}},\label{eq:max-ent-state}%
\end{equation}
with $\{|i\rangle_{\hat{A}}\}_{i}$ and $\{|i\rangle_{\hat{B}}\}_{i}$
orthonormal bases. The (asymptotic)\ PPT\ exact entanglement cost of
$\rho_{AB}$ is defined as
\begin{equation}
E_{\operatorname{PPT}}(\rho_{AB}):=\limsup_{n\rightarrow\infty}\frac{1}%
{n}E_{\operatorname{PPT}}^{(1)}(\rho_{AB}^{\otimes n}%
).\label{eq:-asymp-PPT-cost}%
\end{equation}

By building on earlier results from \cite{Audenaert2003,Matthews2008}, our
first result is for the one-shot PPT\ exact
entanglement cost:%
\begin{proposition}\label{eq:PPT oneshot}
For a given bipartite quantum state $\rho_{AB}$, its one-shot PPT\ exact
entanglement cost is given by
\begin{multline}
E_{\operatorname{PPT}}^{(1)}(\rho_{AB})=\\
\inf\left\{
\begin{array}
[c]{c}%
\log_{2}m:G_{AB}\geq0,\ \operatorname{Tr}[G_{AB}]=1,\\
-\left(  m-1\right)  G_{AB}^{T_{B}}\leq\rho_{AB}^{T_{B}}\leq\left(
m+1\right)  G_{AB}^{T_{B}}, m\in\mathbb{N} %
\end{array}
\right\}  .
\end{multline}
\end{proposition}

The proof of this result involves an achievability and optimality part. The
achievability part constructs the channel $\Lambda_{\hat{A}\hat{B}\rightarrow
AB}\in\operatorname{PPT}$ as the following measure-prepare procedure:%
\begin{multline}
\Lambda_{\hat{A}\hat{B}\rightarrow AB}(\omega_{\hat{A}\hat{B}}):=\rho
_{AB}\operatorname{Tr}[\Phi_{\hat{A}\hat{B}}^{m}\omega_{\hat{A}\hat{B}}]\\
+G_{AB}\operatorname{Tr}[\left(  I_{\hat{A}\hat{B}}-\Phi_{\hat{A}\hat{B}}%
^{m}\right)  \omega_{\hat{A}\hat{B}}],
\end{multline}
for $G_{AB}$ a quantum state satisfying $-\left(  m-1\right)  G_{AB}^{T_{B}%
}\leq\rho_{AB}^{T_{B}}\leq\left(  m+1\right)  G_{AB}^{T_{B}}$. That
$\Lambda_{\hat{A}\hat{B}\rightarrow AB}$ is a quantum channel follows
immediately from its construction, and that $\Lambda_{\hat{A}\hat
{B}\rightarrow AB}\in\operatorname{PPT}$ follows from the constraint on
$G_{AB}$. For the optimality part, we exploit the symmetry of the maximally
entangled state $\Phi_{\hat{A}\hat{B}}^{d}$, that it is invariant under the
unitary channel $\left(  U\otimes\overline{U}\right)  \left(  \cdot\right)  \left(
U\otimes\overline{U}\right)  ^{\dag}$ for an arbitrary unitary~$U$, in order
to constrain the set of channels that we have to consider for the PPT\ exact
entanglement cost. Then by applying the constraint that $\Lambda_{\hat{A}%
\hat{B}\rightarrow AB}\in\operatorname{PPT}$, it follows that the
constructed channel is optimal.

\textit{$\kappa$-entanglement.}---The bottleneck of solving the PPT entanglement cost of a general bipartite state lies in determining the regularization of the one-shot cost, which involves evaluating the limit of a series of optimization problems. To overcome this difficulty, we introduce an  efficiently computable entanglement measure, called \textit{$\kappa$-entanglement}, defined as%
\begin{equation*}
E_{\kappa}(\rho_{AB}):=\log_{2}\inf_{S_{AB}\geq0}\left\{  \operatorname{Tr}%
[S_{AB}]:-S_{AB}^{T_{B}}\leq\rho_{AB}^{T_{B}}\leq S_{AB}^{T_{B}}\right\}  .
\end{equation*}
In particular, $E_{\kappa}$ can be computed by means of a semi-definite program (SDP)~\cite{Vandenberghe1996} (see Section \ref{app:SDP} of \cite{ECappendix} for details). SDPs can be computed efficiently by polynomial-time algorithms~\cite{Khachiyan1980,AHK05,AK07,AHK12,LSW15} and are often applied in quantum information (e.g., \cite{Fletcher2007,Watrous2009,Leung2015c,Lami2018,Fang2018,XWthesis,Wang2018,WWW19}).
The CVX software~\cite{Grant2008} allows one to compute SDPs in practice.

By observing that $\kappa$-entanglement is a relaxation of the one-shot cost in~Proposition~\ref{eq:PPT oneshot} up to small corrections, we arrive
at the following bounds on the one-shot exact entanglement cost:%
\begin{proposition}\label{prop:one-shot}
For a bipartite state $\rho_{AB}$, we have
\begin{equation}
\log_{2}\!\left(  2^{E_{\kappa}(\rho_{AB})}-1\right)  \leq E_{\operatorname{PPT}%
}^{(1)}(\rho_{AB})\leq\log_{2}\!\left(  2^{E_{\kappa}(\rho_{AB})}+2\right)
.\label{eq:one-shot-cost-kappa-ent}%
\end{equation}
\end{proposition}

This result gives a tight and efficiently computable bound for the one-shot PPT exact entanglement cost in terms of $\kappa$-entanglement. A rigorous proof can be found in \cite{ECappendix}. Thus, the inequality in Proposition~\ref{prop:one-shot} demonstrates
that\ the $\kappa$-entanglement is  closely related to the one-shot
PPT\ exact entanglement cost, and as both the operational quantity
$E_{\operatorname{PPT}}^{(1)}(\rho_{AB})$ and the entanglement measure
$E_{\kappa}(\rho_{AB})$ become larger, the gap between them disappears.

In addition to being efficiently calculable by means of a semi-definite program, the $\kappa$-entanglement possesses several properties desirable for an
entanglement measure, including monotonicity under selective
C-PPT-P operations, additivity, faithfulness, and normalization. We elaborate
on each of these briefly now. The monotonicity is the following inequality:%
\begin{equation}
E_{\kappa}(\rho_{AB})\geq\sum_{x:p(x)>0}p(x)E_{\kappa}(\rho_{A^{\prime
}B^{\prime}}^{x}),
\label{eq:monotone-C-PPT-P}
\end{equation}
where $p(x):=\operatorname{Tr}[\mathcal{P}_{AB\rightarrow A^{\prime}B^{\prime
}}^{x}(\rho_{AB})]$, the set $\{\mathcal{P}_{AB\rightarrow A^{\prime}%
B^{\prime}}^{x}\}_{x}$ consists of completely positive, trace non-increasing,
{C-PPT-P} maps such that $\sum_{x}\mathcal{P}_{AB\rightarrow A^{\prime}%
B^{\prime}}^{x}$ is trace preserving, and $\rho_{A^{\prime}B^{\prime}}%
^{x}:=\mathcal{P}_{AB\rightarrow A^{\prime}B^{\prime}}^{x}(\rho_{AB})/p(x)$.
The inequality in \eqref{eq:monotone-C-PPT-P} asserts that $\kappa$-entanglement does not increase on
average under the action of selective C-PPT-P operations, which include selective LOCC operations as a special case. Additivity is the following statement, which is critical
for establishing one of the key results of our paper:%
\begin{equation}
E_{\kappa}(\omega_{A_{1}A_{2}:B_{1}B_{2}})=E_{\kappa}(\rho_{A_{1}B_{1}%
})+E_{\kappa}(\theta_{A_{2}B_{2}}),\label{eq:additivity-states}%
\end{equation}
where $\omega_{A_{1}A_{2}:B_{1}B_{2}}:=\rho_{A_{1}B_{1}}\otimes\theta
_{A_{2}B_{2}}$ and $\rho_{A_{1}B_{1}}$ and $\theta_{A_{2}B_{2}}$ are quantum
states. Faithfulness is that $E_{\kappa}(\rho_{AB})=0$ if and only if
$\rho_{AB}$ is a PPT\ state. Finally, normalization is that $E_{\kappa}(\Phi_{AB}%
^{d})=\log_{2}d$ for $\Phi_{AB}^{d}$ a maximally entangled state of the form
in \eqref{eq:max-ent-state}. Proofs of the properties above are provided in~\cite{ECappendix}.

\textit{Exact entanglement cost.}---The PPT exact entanglement cost $\EPPT$ has been a longstanding open question  since it was first introduced in \cite{Audenaert2003}. The previously best known upper and lower bounds \cite{Audenaert2003} are tight for  general Werner states, but they are not tight in general. The difficulty of determining $\EPPT$ comes from the fact that the one-shot cost is not an SDP, and its regularization makes the problem more intractable. However, by utilizing the techniques of semi-definite optimization and relaxation, we prove that the asymptotic exact entanglement cost of a state $\rho_{AB}$ is given by $\EK(\rho_{AB})$. Specifically, by exploiting \eqref{eq:one-shot-cost-kappa-ent}, the definition of PPT\ exact entanglement cost in \eqref{eq:-asymp-PPT-cost}, and the additivity of $\kappa$-entanglement in \eqref{eq:additivity-states}, we
arrive at one of our core contributions:%
\begin{theorem}\label{thm:PPT cost}
The PPT\ exact entanglement cost of an arbitrary bipartite state $\rho_{AB}$ is given by
\begin{equation}
E_{\operatorname{PPT}}(\rho_{AB})=E_{\kappa}(\rho_{AB}%
).\label{eq:operational-main-states-kappa}%
\end{equation}
\end{theorem}

This result has two important consequences. First,
it demonstrates that $\kappa$-entanglement precisely determines the PPT\ exact entanglement cost of an arbitrary quantum state. Notably, this is the first time that an entanglement measure for general bipartite states has been proven not only to possess a direct operational meaning but also to be efficiently computable, thus solving a question that has remained open since the inception of entanglement theory over two decades ago. 
Second, note that $\EK$ is additive (cf., Eq.~\eqref{eq:additivity-states}), so that Theorem~\ref{thm:PPT cost} implies that the PPT exact entanglement cost is additive in general:
 \begin{equation}
 \EPPT(\rho_{AB} \otimes \omega_{A'B'}) = 
  \EPPT(\rho_{AB})
  +
   \EPPT(\omega_{A'B'}).
 \end{equation}

Based on Theorem~\ref{thm:PPT cost}, we further show that the PPT exact entanglement cost violates the convexity and monogamy inequalities, which gives insight to the fundamental structure of entanglement.
Recall that for an entanglement measure~$E$,
convexity is the following statement:%
\begin{equation}
E(\overline{\rho}_{AB})\leq\sum_{z}p(z)E(\rho_{AB}^{z}),
\end{equation}
where $p(z)$ is a probability distribution, $\left\{  \rho_{AB}^{z}\right\}
_{z}$ is a set of states, and $\overline{\rho}_{AB}:=\sum_{z}p(z)\rho_{AB}%
^{z}$. This is not true for the PPT exact entanglement cost. 
In particular, let us choose the two-qubit states $\rho_1:=\Phi_2$,  $\rho_2:=\frac{1}{2}(\proj{00}+\proj{11})$, and their average $\rho:=\frac{1}{2} (\rho_1+\rho_2)$.
By direct calculation, we find that
$\EPPT(\rho_1)=1$,
$\EK(\rho_2)=0$, and
$\EK(\rho)=\log_2 \frac{3}{2}$,
from which we conclude that
\begin{align}
\EK(\rho) 	> \frac 1 2(\EK(\rho_1)+\EK(\rho_2)),
\end{align}
This implies the following:
\begin{proposition}[No convexity]
\label{prop:no-convexity}
The PPT exact entanglement cost is not generally convex.
\end{proposition}

As a consequence of the finding above, the exact entanglement cost of preparing the average of two states $\rho_1$ and $\rho_2$ can sometimes be strictly larger than the average exact entanglement cost of preparing each state separately. Convexity is sometimes associated with the loss of entanglement under the discarding of classical information. However, this is only sensible for entanglement measures that obey what is known as the ``flags'' property \cite{MH05,Plenio2005b,Horodecki2009a}. Note that the $\kappa$-entanglement does not possess this property (if it were to, then it would be convex). We stress here that the $\kappa$-entanglement is monotone under LOCC, as indicated in \eqref{eq:monotone-C-PPT-P}, which implies that it does not increase when Alice and Bob discard local registers in their possession. Since local registers of course can be classical registers, we conclude that $\kappa$-entanglement does not increase under the loss of classical information in this sense. The lack of convexity for $\kappa$-entanglement simply means that in some cases, the cost of preparing the average of two states can exceed the average cost of preparing the individual states. See \cite{Plenio2005b} for further discussions about this point.

Monogamy of an entanglement measure $E$\ is as follows~\cite{T04}:%
\begin{equation}
E(\rho_{A:BC})\geq E(\rho_{A:B})+E(\rho_{A:C}),
\end{equation}
where $\rho_{ABC}$ is a tripartite state. It captures the idea that
the sum of the entanglement that Alice shares individually with Bob and
Charlie when they are all in separate laboratories cannot exceed the
entanglement that she has with them when Bob and Charlie are in the same
laboratory. Here, by utilizing $\kappa$-entanglement, we show that $E_{\operatorname{PPT}}(\psi_{AB}) + E_{\operatorname{PPT}}(\psi_{AC}) > E_{\operatorname{PPT}}(\psi_{A(BC)})$ for the tripartite state ${\ket\psi}_{ABC}=\frac{1}{2}(\ket{000}_{ABC}+\ket{011}_{ABC}+\sqrt 2 \ket {110}_{ABC})$. Thus, we have the following:
\begin{proposition}[No monogamy]
The PPT exact entanglement cost is not generally monogamous.
\end{proposition}

In some literature on entanglement (see, e.g., \cite{Plenio2005b}), the properties of convexity and monogamy were thought to be essential features of entanglement, but the fact that $\kappa$-entanglement is neither convex nor monogamous, while having a clear-cut operational meaning, calls into question whether these
properties are really necessary for an entanglement measure. See~\cite{Plenio2005b,Gour2018monogamyof}\ for other discussions questioning the necessity of these two properties.

As another implication of our results, we find by example that exact PPT entanglement manipulation is irreversible. In particular, this example together with several classes of examples in Section \ref{sec:examples-states irreversible} of \cite{ECappendix} imply that $\EPPT$ is generally not equal to the logarithmic negativity $E_N$~\cite{Vidal2002,Plenio2005b}.  Consider the following rank-two state supported on the $3\times 3$ antisymmetric subspace \cite{Wang2016d}:
\begin{equation}
\rho^{v}_{AB}=\frac{1}{2}(\proj{v_1}_{AB}+\proj{v_2}_{AB})
\end{equation}
with 
\begin{align}
\ket {v_1}_{AB} & := (\ket {01}_{AB}-\ket{10}_{AB})/{\sqrt 2},\\
 \ket {v_2}_{AB} & := (\ket {02}_{AB}-\ket{20}_{AB})/{\sqrt 2}.
 \end{align}
 For the state $\rho^{v}_{AB}$, 
it holds that
\begin{multline}
E_N(\rho^{v}_{AB})
=\log_2\! \left(1+{1}/{\sqrt 2}\right)
 < \EPPT(\rho^v_{AB})=1 \\
 < \log_2 Z(\rho^v_{AB})= 
\log_2 \!\left(1+{13}/{4\sqrt 2}\right),
\end{multline}
where $\log_2 Z(\rho_{AB})$ is the previous upper bound on $\EPPT$ from \cite{Audenaert2003}. The strict inequalities above also imply that the previously best known lower and upper bounds from~\cite{Audenaert2003} are not tight. Since the logarithmic negativity is known to be an upper bound on PPT exact distillable entanglement \cite{HHH00a,Vidal2002}, we conclude that exact PPT entanglement manipulation is irreversible.

\textit{Conclusions.}---%
We have shown that the PPT exact entanglement cost is equal to the $\kappa$-entanglement, a single-letter, efficiently computable entanglement measure.
Our results constitute a significant development for entanglement theory, representing the first time that an entanglement measure has been proven to be not only efficiently computable but also to possess a direct information-theoretic meaning. Prior to our work, every other entanglement measure introduced previously possesses only one of these two properties, and thus they were either not accessible computationally or not information-theoretically meaningful. Our work closes this outstanding theoretical gap, because our entanglement measure can be calculated efficiently by semi-definite programming and it has an operational meaning as the cost of maximally entangled states needed to prepare a state. This unique feature improves our understanding of the fundamental structure and power of entanglement. 

Furthermore, we have shown that the $\kappa$-entanglement (or exact PPT entanglement cost) possesses properties such as additivity, monotonicity, faithfulness, normalization, non-convexity, and non-monogamy. These results give insight into the structure of quantum entanglement that have not previously been observed in a general operational setting and bring a significant simplification to entanglement theory. 
In particular, most prior discussions about the structure and properties of entanglement are based on entanglement measures. However, none of these measures, with the exception of the regularized relative entropy of entanglement, possesses a direct operational meaning. Thus, the connection made by Theorem~\ref{thm:PPT cost} allows for the study of the structure of entanglement via an entanglement measure possessing a direct operational meaning. Given that $E_\kappa = \EPPT$ is neither convex nor monogamous, this raises questions of whether these properties should really be required or necessary for measures of entanglement, in contrast to the discussions put forward in \cite{T04,Horodecki2009a} based on intuition.

Our results may also shed light on the open question of whether distillable entanglement is convex \cite{Shor2001}, but this remains the topic of future work. In the multi-partite setting, it is known that a version of distillable entanglement is not convex \cite{Shor2003}.

\textbf{Acknowledgements}:\ We are grateful to Renato Renner and Andreas
Winter for insightful discussions. Part of this work was done when XW was at the University of Maryland. MMW acknowledges support from the National Science Foundation
under Award Nos.~1350397 and 1907615.


%

\clearpage
\appendix
\onecolumngrid
\begin{center}
\vspace*{.5\baselineskip}
{\textbf{\large Supplemental Material: \\[3pt] The cost of quantum entanglement simplified}}\\[1pt] \quad \\
\end{center}

%

\renewcommand{\theequation}{S\arabic{equation}}
\renewcommand{\thetheorem}{S\arabic{theorem}}
\setcounter{equation}{0}
\setcounter{figure}{0}
\setcounter{table}{0}
\setcounter{section}{0}

\large

This supplementary material provides a more detailed analysis and proofs of the results stated in the main text. On occasion, we reiterate some of the steps in the main text in order to make the supplementary material more clear and self contained.

 \section{One-shot PPT exact entanglement cost}

Let $\Omega$ represent a set of free channels. Examples of interest include the set of \text{LOCC} channels or the set of completely-\text{PPT}-preserving channels. 
The one-shot exact entanglement cost of a bipartite state $\rho_{AB}$, under the $\Omega$ channels, is defined as
\begin{align}
E^{(1)}_{\Omega}(\rho_{AB})= \inf_{d\in \mathbb{N}, \Lambda\in \Omega}\left\{\log_2 d:   \rho_{AB}=\Lambda_{\hat{A}\hat{B}\to AB} (\Phi^{d}_{\hat{A}\hat{B}})\right\},
\end{align}
where $\mathbb{N}:= \{1,2,3,\ldots\}$ and $\Phi^{d}_{\hat{A}\hat{B}}=[1/d]\sum_{i,j=1}^{d}\ketbra{ii}{jj}_{\hat{A}\hat{B}}$ represents the standard maximally entangled state of Schmidt rank~$d$. The exact entanglement cost of a bipartite state $\rho_{AB}$, under the $\Omega$ channels, is defined as
\begin{align}
E_{\Omega}(\rho_{AB})= \limsup_{n \to \infty} \frac{1}{n}E^{(1)}_{\Omega}(\rho_{AB}^{\ox n}).
\label{eq:asympt-cost}
\end{align}
The exact entanglement cost under LOCC channels was previously considered in \cite{N99,TH00,H06book,YC18}, while the exact entanglement cost under completely-PPT-preserving channels was considered in \cite{Audenaert2003,Matthews2008}.

The reason for the appearance of the limit superior in \eqref{eq:asympt-cost} is as follows. The definition in \eqref{eq:asympt-cost} involves the sequence $\{E_n\}_n$ of non-negative reals
$E_n := \frac{1}
{n}E_{\operatorname{PPT}}^{(1)}(\rho_{AB}^{\otimes n}
)$. The limit of this sequence does not necessarily exist \textit{a priori}, but the limit inferior and limit superior always exist for any sequence. Thus, we should decide which  of these two possibilities is appropriate for the entanglement cost problem. It is sensible that the asymptotic cost should be sufficient to  cover the entanglement needs of all but finitely many terms in the sequence. Given this requirement, the limit superior is the appropriate limiting notion here. However, as we shall see in what follows, when the set $\Omega$ is the set of completely-\text{PPT}-preserving channels, the limit superior and limit inferior are actually equal and given by the $\kappa$-entanglement (defined in the main text and later on in Definition~\ref{def:kappa-ent}).

In \cite{Audenaert2003}, the following bounds were given for $\EPPT$:
\begin{align}
E_N(\rho_{AB})\le	\EPPT(\rho_{AB})\le \log_2  Z(\rho_{AB}),
\label{eq:ape-bnds}
\end{align}
the lower bound being the logarithmic negativity \cite{Vidal2002,Plenio2005b}, defined as  \begin{equation}
  E_{N}(\rho_{AB})\coloneqq\log_2\left \Vert \rho_{AB}^{T_{B}}\right\Vert_1,
  \label{eq:log-neg}
\end{equation}
and the upper bound defined in terms of
\begin{equation}
Z(\rho_{AB})\coloneqq \left\|\rho_{AB}^{T_B}\right\|_1 +\dim(\rho_{AB})\max\{0,-\lambda_{\min}(|\rho_{AB}^{T_B}|^{T_B})\}.
\end{equation}
Due to the presence of the dimension factor $\dim(\rho_{AB})$, the upper bound in \eqref{eq:ape-bnds} clearly only applies in the case that $\rho_{AB}$ is finite-dimensional.

In what follows, we first recast
$E^{(1)}_{\PPT}(\rho_{AB})$ as an optimization problem, by building on previous developments in \cite{Audenaert2003,Matthews2008}. After that, we bound
$E^{(1)}_{\PPT}(\rho_{AB})$
in terms of $E_\kappa$, by observing that $E_\kappa$ is a relaxation of the optimization problem for $E^{(1)}_{\PPT}(\rho_{AB})$. We then finally prove that $E_{\PPT}(\rho_{AB})$ is equal to $E_\kappa$. 

\begin{proposition}
\label{prop:exact-cost-states}
Let $\rho_{AB}$ be a bipartite state acting on a separable Hilbert space. Then the one-shot exact PPT-entanglement cost $E_{\operatorname{PPT}}^{(1)}%
(\rho_{AB})$ is given by the following
optimization:%
\begin{equation}
E_{\operatorname{PPT}}^{(1)}(\rho_{AB})=\inf_{m\in\mathbb{N}}\left\{  \log_{2}m:-\left(
m-1\right)  G_{AB}^{T_{B}}\leq\rho_{AB}^{T_{B}}\leq\left(  m+1\right)
G_{AB}^{T_{B}},\ G_{AB}\geq0,\ \operatorname{Tr}G_{AB}=1\right\}  .
\label{eq:op-quantity-PPT-cost}
\end{equation}
\end{proposition}
\begin{proof}
The achievability part features a construction of a completely-PPT-preserving
channel $\mathcal{P}_{\hat{A}\hat{B}\rightarrow AB}$ such that $\mathcal{P}%
_{\hat{A}\hat{B}\rightarrow AB}(\Phi_{\hat{A}\hat{B}}^{m})=\rho_{AB}$, and
then the converse part demonstrates that the constructed channel is
essentially the only form that is needed to consider for the one-shot exact
PPT-entanglement cost task. The achievability part directly employs some insights of \cite{Audenaert2003}, while the converse part directly employs insights of \cite{Matthews2008}. In what follows, we give a proof for the sake of completeness.

Let $m\geq1$ be a positive integer and $G_{AB}$ a density operator such that
the following inequalities hold%
\begin{equation}
-\left(  m-1\right)  G_{AB}^{T_{B}}\leq\rho_{AB}^{T_{B}}\leq\left(
m+1\right)  G_{AB}^{T_{B}}.\label{eq:condition-for-PPT-preserve}%
\end{equation}
Note that the inequalities in \eqref{eq:condition-for-PPT-preserve} imply the following inequality
\begin{equation}
-\left(  m-1\right)  G_{AB}^{T_{B}} \leq\left(
m+1\right)  G_{AB}^{T_{B}} 
\quad \Longleftrightarrow \quad 0  \leq \left(
m+1\right)  G_{AB}^{T_{B}} + \left(  m-1\right)  G_{AB}^{T_{B}} = m G_{AB}^{T_{B}},
\end{equation}
which in turn implies that $G_{AB}^{T_{B}} \geq 0$, so that $G_{AB}$ is a PPT state.

Then we take the completely-PPT-preserving channel $\mathcal{P}_{\hat{A}%
\hat{B}\rightarrow AB}$ to be as follows:%
\begin{equation}
\mathcal{P}_{\hat{A}\hat{B}\rightarrow AB}(X_{\hat{A}\hat{B}})=\rho
_{AB}\operatorname{Tr}[\Phi_{\hat{A}\hat{B}}^{m}X_{\hat{A}\hat{B}}%
]+G_{AB}\operatorname{Tr}[(\1_{\hat{A}\hat{B}}-\Phi_{\hat{A}\hat{B}}%
^{m})X_{\hat{A}\hat{B}}].
\end{equation}
The action of $\mathcal{P}_{\hat{A}\hat{B}\rightarrow AB}$ can be understood
as a measure-prepare channel (and is thus a channel):\ first perform the
measurement $\{\Phi_{\hat{A}\hat{B}}^{m},\1_{\hat{A}\hat{B}}-\Phi_{\hat{A}%
\hat{B}}^{m}\}$, and if the outcome $\Phi_{\hat{A}\hat{B}}^{m}$ occurs,
prepare the state $\rho_{AB}$, and otherwise, prepare the state $G_{AB}$. To
see that the channel $\mathcal{P}_{\hat{A}\hat{B}\rightarrow AB}$ is a
completely-PPT-preserving channel, we now verify that the map $T_{B}%
\circ\mathcal{P}_{\hat{A}\hat{B}\rightarrow AB}\circ T_{\hat{B}}$ is
completely positive. Let $Y_{R_{A}\hat{A}\hat{B}R_{B}}$ be a positive
semi-definite operator with $R_{A}$ isomorphic to $\hat{A}$ and $R_{B}$
isomorphic to~$\hat{B}$. Then consider that%
\begin{align}
&  (T_{B}\circ\mathcal{P}_{\hat{A}\hat{B}\rightarrow AB}\circ T_{\hat{B}%
})(Y_{R_{A}\hat{A}\hat{B}R_{B}})\nonumber\\
&  =\rho_{AB}^{T_{B}}\operatorname{Tr}_{\hat{A}\hat{B}}[\Phi_{\hat{A}\hat{B}%
}^{m}Y_{R_{A}\hat{A}\hat{B}R_{B}}^{T_{\hat{B}}}]+G_{AB}^{T_{B}}%
\operatorname{Tr}_{\hat{A}\hat{B}}[(\1_{\hat{A}\hat{B}}-\Phi_{\hat{A}\hat{B}%
}^{m})Y_{R_{A}\hat{A}\hat{B}R_{B}}^{T_{\hat{B}}}]\\
&  =\rho_{AB}^{T_{B}}\operatorname{Tr}_{\hat{A}\hat{B}}[(\Phi_{\hat{A}\hat{B}%
}^{m})^{T_{\hat{B}}}Y_{R_{A}\hat{A}\hat{B}R_{B}}]+G_{AB}^{T_{B}}%
\operatorname{Tr}_{\hat{A}\hat{B}}[(\1_{\hat{A}\hat{B}}-\Phi_{\hat{A}\hat{B}%
}^{m})^{T_{\hat{B}}}Y_{R_{A}\hat{A}\hat{B}R_{B}}]\\
&  =\frac{\rho_{AB}^{T_{B}}}{m}\operatorname{Tr}_{\hat{A}\hat{B}}[F_{\hat
{A}\hat{B}}Y_{R_{A}\hat{A}\hat{B}R_{B}}]+G_{AB}^{T_{B}}\operatorname{Tr}%
_{\hat{A}\hat{B}}[(\1_{\hat{A}\hat{B}}-F_{\hat{A}\hat{B}}/m)Y_{R_{A}\hat{A}%
\hat{B}R_{B}}]\\
&  =\frac{\rho_{AB}^{T_{B}}}{m}\operatorname{Tr}_{\hat{A}\hat{B}}[F_{\hat
{A}\hat{B}}Y_{R_{A}\hat{A}\hat{B}R_{B}}]+\frac{G_{AB}^{T_{B}}}{m}%
\operatorname{Tr}_{\hat{A}\hat{B}}[(m\1_{\hat{A}\hat{B}}-F_{\hat{A}\hat{B}%
})Y_{R_{A}\hat{A}\hat{B}R_{B}}]\\
&  =\frac{\rho_{AB}^{T_{B}}}{m}\operatorname{Tr}_{\hat{A}\hat{B}}[(\Pi
_{\hat{A}\hat{B}}^{\mathcal{S}}-\Pi_{\hat{A}\hat{B}}^{\mathcal{A}}%
)Y_{R_{A}\hat{A}\hat{B}R_{B}}]\nonumber\\
&  \qquad+\frac{G_{AB}^{T_{B}}}{m}\operatorname{Tr}_{\hat{A}\hat{B}}%
[(m(\Pi_{\hat{A}\hat{B}}^{\mathcal{S}}+\Pi_{\hat{A}\hat{B}}^{\mathcal{A}%
})-(\Pi_{\hat{A}\hat{B}}^{\mathcal{S}}-\Pi_{\hat{A}\hat{B}}^{\mathcal{A}%
}))Y_{R_{A}\hat{A}\hat{B}R_{B}}]\\
&  =\frac{1}{m}\left[  \rho_{AB}^{T_{B}}+\left(  m-1\right)  G_{AB}^{T_{B}%
}\right]  \operatorname{Tr}_{\hat{A}\hat{B}}[\Pi_{\hat{A}\hat{B}}%
^{\mathcal{S}}Y_{R_{A}\hat{A}\hat{B}R_{B}}]\nonumber\\
&  \qquad+\frac{1}{m}\left[  \left(  m+1\right)  G_{AB}^{T_{B}}-\rho
_{AB}^{T_{B}}\right]  \operatorname{Tr}_{\hat{A}\hat{B}}[\Pi_{\hat{A}\hat{B}%
}^{\mathcal{A}}Y_{R_{A}\hat{A}\hat{B}R_{B}}].
\end{align}
The third equality follows because the partial transpose of $\Phi_{\hat{A}%
\hat{B}}^{m}$ is proportional to the unitary flip or swap operator $F_{\hat{A}\hat
{B}}$. The fifth equality follows by recalling the definition of the
projections onto the symmetric and antisymmetric subspaces respectively as%
\begin{equation}
\Pi_{\hat{A}\hat{B}}^{\mathcal{S}}=\frac{\1_{\hat{A}\hat{B}}+F_{\hat{A}\hat{B}%
}}{2},\qquad\Pi_{\hat{A}\hat{B}}^{\mathcal{A}}=\frac{\1_{\hat{A}\hat{B}%
}-F_{\hat{A}\hat{B}}}{2}.
\end{equation}
As a consequence of the condition in \eqref{eq:condition-for-PPT-preserve}, it
follows that $T_{B}\circ\mathcal{P}_{\hat{A}\hat{B}\rightarrow AB}\circ
T_{\hat{B}}$ is completely positive, so that $\mathcal{P}_{\hat{A}\hat
{B}\rightarrow AB}$ is a completely-PPT-preserving channel as claimed. In
fact, we can see that $T_{B}\circ\mathcal{P}_{\hat{A}\hat{B}\rightarrow
AB}\circ T_{\hat{B}}$ is a measure-prepare channel:\ first perform the
measurement $\{\Pi_{\hat{A}\hat{B}}^{\mathcal{S}},\Pi_{\hat{A}\hat{B}%
}^{\mathcal{A}}\}$ and if the outcome $\Pi_{\hat{A}\hat{B}}^{\mathcal{S}}$
occurs, prepare the state $\frac{1}{m}[\rho_{AB}^{T_{B}}+\left(  m-1\right)
G_{AB}^{T_{B}}]$, and otherwise, prepare the state $\frac{1}{m}[\left(
m+1\right)  G_{AB}^{T_{B}}-\rho_{AB}^{T_{B}}]$. Thus, it follows that
$\mathcal{P}_{\hat{A}\hat{B}\rightarrow AB}$ accomplishes the one-shot exact
PPT-entanglement cost task, in the sense that%
\begin{equation}
\mathcal{P}_{\hat{A}\hat{B}\rightarrow AB}(\Phi_{\hat{A}\hat{B}}^{m}%
)=\rho_{AB}.
\end{equation}
By taking an infimum over all $m$ and density operators $G_{AB}$ such that
\eqref{eq:condition-for-PPT-preserve} holds, it follows that the quantity on
the right-hand side of \eqref{eq:op-quantity-PPT-cost} is greater than
or equal to $E_{\operatorname{PPT}}^{(1)}(\rho_{AB})$.

Now we prove the opposite inequality. Let $\mathcal{P}_{\hat{A}\hat
{B}\rightarrow AB}$ denote an arbitrary completely-PPT-preserving channel such
that%
\begin{equation}
\mathcal{P}_{\hat{A}\hat{B}\rightarrow AB}(\Phi_{\hat{A}\hat{B}}^{m}%
)=\rho_{AB}.\label{eq:ppt-trans-to-rho}%
\end{equation}
Let $\mathcal{T}_{\hat{A}\hat{B}}$ denote the following isotropic twirling
channel \cite{W89,Horodecki99,Watrous2011b}:%
\begin{align}
\mathcal{T}_{\hat{A}\hat{B}}(X_{\hat{A}\hat{B}}) &  =\int dU\ (U_{\hat{A}%
}\otimes\overline{U}_{\hat{B}})X_{\hat{A}\hat{B}}(U_{\hat{A}}\otimes
\overline{U}_{\hat{B}})^{\dag}\\
&  =\Phi_{\hat{A}\hat{B}}^{m}\operatorname{Tr}[\Phi_{\hat{A}\hat{B}}%
^{m}X_{\hat{A}\hat{B}}]+\frac{\1_{\hat{A}\hat{B}}-\Phi_{\hat{A}\hat{B}}^{m}%
}{m^{2}-1}\operatorname{Tr}[(\1_{\hat{A}\hat{B}}-\Phi_{\hat{A}\hat{B}}%
^{m})X_{\hat{A}\hat{B}}].
\end{align}
The channel $\mathcal{T}_{\hat{A}\hat{B}}$ is an LOCC channel, and thus is
completely-PPT-preserving. Furthermore, due to the fact that $\mathcal{T}%
_{\hat{A}\hat{B}}(\Phi_{\hat{A}\hat{B}}^{m})=\Phi_{\hat{A}\hat{B}}^{m}$, it
follows that%
\begin{equation}
(\mathcal{P}_{\hat{A}\hat{B}\rightarrow AB}\circ\mathcal{T}_{\hat{A}\hat{B}%
})(\Phi_{\hat{A}\hat{B}}^{m})=\rho_{AB}.
\end{equation}
Thus, for any completely-PPT-preserving channel $\mathcal{P}_{\hat{A}\hat
{B}\rightarrow AB}$ such that \eqref{eq:ppt-trans-to-rho} holds, there exists
another channel $\mathcal{P}_{\hat{A}\hat{B}\rightarrow AB}^{\prime
}\coloneqq\mathcal{P}_{\hat{A}\hat{B}\rightarrow AB}\circ\mathcal{T}_{\hat{A}\hat{B}%
}$ achieving the same performance, and so it suffices to focus on the channel
$\mathcal{P}_{\hat{A}\hat{B}\rightarrow AB}^{\prime}$ in order to establish an
expression for the one-shot exact PPT-entanglement cost. Then, consider that,
for any input state $\tau_{\hat{A}\hat{B}}$, we have that%
\begin{align}
&  \mathcal{P}_{\hat{A}\hat{B}\rightarrow AB}^{\prime}(\tau_{\hat{A}\hat{B}%
})\nonumber\\
&  =\mathcal{P}_{\hat{A}\hat{B}\rightarrow AB}\!\left(  \Phi_{\hat{A}\hat{B}%
}^{m}\operatorname{Tr}[\Phi_{\hat{A}\hat{B}}^{m}\tau_{\hat{A}\hat{B}}%
]+\frac{\1_{\hat{A}\hat{B}}-\Phi_{\hat{A}\hat{B}}^{m}}{m^{2}-1}%
\operatorname{Tr}[(\1_{\hat{A}\hat{B}}-\Phi_{\hat{A}\hat{B}}^{m})\tau_{\hat
{A}\hat{B}}]\right)  \\
&  =\mathcal{P}_{\hat{A}\hat{B}\rightarrow AB}(\Phi_{\hat{A}\hat{B}}%
^{m})\operatorname{Tr}[\Phi_{\hat{A}\hat{B}}^{m}\tau_{\hat{A}\hat{B}%
}]+\mathcal{P}_{\hat{A}\hat{B}\rightarrow AB}\!\left(  \frac{\1_{\hat{A}\hat{B}%
}-\Phi_{\hat{A}\hat{B}}^{m}}{m^{2}-1}\right)  \operatorname{Tr}[(\1_{\hat
{A}\hat{B}}-\Phi_{\hat{A}\hat{B}}^{m})\tau_{\hat{A}\hat{B}}]\\
&  =\rho_{AB}\operatorname{Tr}[\Phi_{\hat{A}\hat{B}}^{m}\tau_{\hat{A}\hat{B}%
}]+G_{AB}\operatorname{Tr}[(\1_{\hat{A}\hat{B}}-\Phi_{\hat{A}\hat{B}}^{m}%
)\tau_{\hat{A}\hat{B}}],
\end{align}
where we have set%
\begin{equation}
G_{AB}=\mathcal{P}_{\hat{A}\hat{B}\rightarrow AB}\!\left(  \frac{\1_{\hat{A}%
\hat{B}}-\Phi_{\hat{A}\hat{B}}^{m}}{m^{2}-1}\right)
.\label{eq:G_AB-for-converse-part}%
\end{equation}
In order for $\mathcal{P}_{\hat{A}\hat{B}\rightarrow AB}^{\prime}$ to be
completely-PPT-preserving, it is necessary that $T_{B}\circ\mathcal{P}%
_{\hat{A}\hat{B}\rightarrow AB}\circ T_{\hat{B}}$ is completely positive.
Going through the same calculations as above, we see that it is necessary for
the following operator to be positive semi-definite for an arbitrary positive
semi-definite $Y_{R_{A}\hat{A}\hat{B}R_{B}}$:%
\begin{equation}
\frac{1}{m}\left(  \left[  \rho_{AB}^{T_{B}}+\left(  m-1\right)  G_{AB}%
^{T_{B}}\right]  \operatorname{Tr}_{\hat{A}\hat{B}}[\Pi_{\hat{A}\hat{B}%
}^{\mathcal{S}}Y_{R_{A}\hat{A}\hat{B}R_{B}}]+\left[  \left(  m+1\right)
G_{AB}^{T_{B}}-\rho_{AB}^{T_{B}}\right]  \operatorname{Tr}_{\hat{A}\hat{B}%
}[\Pi_{\hat{A}\hat{B}}^{\mathcal{A}}Y_{R_{A}\hat{A}\hat{B}R_{B}}]\right)  .
\end{equation}
However, since $\Pi_{\hat{A}\hat{B}}^{\mathcal{S}}$ and $\Pi_{\hat{A}\hat{B}%
}^{\mathcal{A}}$ project onto orthogonal subspaces, this is possible only if
the condition in \eqref{eq:condition-for-PPT-preserve} holds for $G_{AB}$
given in \eqref{eq:G_AB-for-converse-part}. Thus, it follows that the quantity
on the right-hand side of \eqref{eq:op-quantity-PPT-cost}\ is less than
or equal to $E_{\operatorname{PPT}}^{(1)}(\rho_{AB})$.
\end{proof}

 \section{$\kappa$-entanglement bounds the one-shot PPT exact entanglement cost}
 

\subsection{$\kappa$-entanglement measure}
Let us first recall the definition of $\kappa$-entanglement.
\begin{definition}[$\kappa$-entanglement measure]
\label{def:kappa-ent}
Let $\rho_{AB}$ be a bipartite state acting on a separable Hilbert space.  The $\kappa$-entanglement measure is defined as follows:
\begin{equation}
\label{eq:a prime}
E_\kappa(\rho_{AB})\coloneqq \inf \{ \log_2 \tr [S_{AB}] : -S_{AB}^{T_B}\le\rho_{AB}^{T_B}\le S_{AB}^{T_B}, \, S_{AB}\ge 0 \}.
\end{equation}
\end{definition}

It is obvious from the definition of $\EPPTone(\rho_{AB})$ that it is equal to zero if $\rho_{AB}$ is a PPT state. As we show later on in Proposition~\ref{prop:faithfulness}, the $\kappa$-entanglement $E_\kappa(\rho_{AB})$ of a state $\rho_{AB}$ is non-negative and equal to zero if and only if $\rho_{AB}$ is a PPT state. Thus, it follows that $\EPPTone(\rho_{AB}) = E_\kappa(\rho_{AB}) = 0 $ if $\rho_{AB}$ is a PPT state.

The following proposition is helpful in estimating $\EPPTone(\rho_{AB})$ in the case that $\rho_{AB}$ is not a PPT state. Due to the aforementioned fact that $E_\kappa(\rho_{AB}) \geq 0$, the lower bound is meaningful even in the trivial case in which $\rho_{AB}$ is a PPT state, in the sense that the lower bound is equal to $-\infty$.

\begin{proposition}
\label{lemma: one-shot}
Let $\rho_{AB}$ be a bipartite state acting on a separable Hilbert space. Then
	\begin{align}
	\log_2 (2^{E_\kappa(\rho_{AB})}-1)\le	\EPPTone(\rho_{AB})\le \log_2 (2^{E_\kappa(\rho_{AB})}+2).
	\label{eq:bounds-for-one-shot-e-kap}
	\end{align}
\end{proposition}

\begin{proof}
To begin with, note that the bounds  trivially hold in the case that $\rho_{AB}$ is a PPT state, with the lower bound being equal to $-\infty$. Thus, we can focus exclusively on the case of an NPT (non-PPT) state such that $E_\kappa(\rho_{AB}) > 0$ and $ \EPPTone(\rho_{AB})\geq \log_2(2) = 1 $. (This point is more important for the proof of the upper bound in \eqref{eq:bounds-for-one-shot-e-kap}.)

The proof of this lemma utilizes basic concepts from optimization theory. 
	Let us first prove the first inequality in \eqref{eq:bounds-for-one-shot-e-kap}. The key idea is to relax the bilinear optimization problem in \eqref{eq:op-quantity-PPT-cost} to a semi-definite optimization problem. When doing so, we find the following:
	\begin{align}
	&\EPPTone(\rho_{AB})\notag \\
	&=\inf \{\log_2 m: -(m-1) G_{AB}^{T_B}\le\rho_{AB}^{T_B}\le (m+1) G_{AB}^{T_B}, \, G_{AB}\ge0, \, \tr G_{AB}=1, m\in\mathbb{N}\}\\
		&\geq \inf \{\log_2 \mu: -(\mu-1) G_{AB}^{T_B}\le\rho_{AB}^{T_B}\le (\mu+1) G_{AB}^{T_B}, \, G_{AB}\ge0, \, \tr G_{AB}=1, \mu \geq 1 \}\\
	& \ge \inf \{\log_2 \mu: -(\mu+1) G_{AB}^{T_B}\le\rho_{AB}^{T_B}\le (\mu+1) G_{AB}^{T_B}, \, G_{AB}\ge0, \, \tr G_{AB}=1, \mu \geq 1\}\\
	& \geq  \inf \{\log_2 (\tr S_{AB}-1): -S_{AB}^{T_B}\le\rho_{AB}^{T_B}\le  S_{AB}^{T_B}, \, S_{AB}\ge0\}\\
	&=\log_2 (2^{E_\kappa(\rho_{AB})}-1).
	\end{align}
	The first inequality follows by removing the integer constraint, so that the optimization is then over a larger set.  The second inequality follows by relaxing the constraint  $-(\mu-1) G_{AB}^{T_B}\le\rho_{AB}^{T_B}$ to $-\left(\mu+1\right) G_{AB}^{T_B}\le\rho_{AB}^{T_B}$, which is possible because $\mu \geq 1$ and
	\begin{align}
	-(\mu-1) G_{AB}^{T_B}\le\rho_{AB}^{T_B}\le (\mu+1) G_{AB}^{T_B}
	& \quad \Rightarrow \quad   -(\mu-1) G_{AB}^{T_B}\le (\mu+1) G_{AB}^{T_B} \\
	& \quad  \Rightarrow \quad   G_{AB}^{T_B} \geq 0.
	\end{align}
	The third inequality follows by  setting $S_{AB}=(\mu+1)G_{AB}$ and noting that $\mu = \operatorname{Tr}[S_{AB}] -1$, as well as the fact that the set $\{S_{AB}\geq 0\}$ is larger than the set $\{(\mu+1)G_{AB} : G_{AB}\geq 0, \operatorname{Tr}[G_{AB}] = 1, \mu\geq 1 \}$. The last equality follows from the definition of $E_\kappa(\rho_{AB})$.
	
	The upper bound is a consequence of the following chain of inequalities:%
\begin{align}
& E_{\text{PPT}}^{(1)}(\rho_{AB}) \notag \\
& =\inf\left\{  \log_{2}m:-\left(
m-1\right)  G_{AB}^{T_{B}}\leq\rho_{AB}^{T_{B}}\leq\left(  m+1\right)
G_{AB}^{T_{B}},\ G_{AB}\geq0,\ \operatorname{Tr}G_{AB}=1,\ m\in\mathbb{N}, m\geq 2%
\right\}  \notag\\
& \leq\inf\left\{  \log_{2}m:-\left(  m-1\right)  G_{AB}^{T_{B}}\leq\rho
_{AB}^{T_{B}}\leq\left(  m-1\right)  G_{AB}^{T_{B}},\ G_{AB}\geq
0,\ \operatorname{Tr}G_{AB}=1,\ m\in\mathbb{N}, m \geq 2\right\}  \notag\\
& =\inf\left\{  \log_{2}\left\lfloor \mu\right\rfloor :-\left(  \left\lfloor
\mu\right\rfloor -1\right)  G_{AB}^{T_{B}}\leq\rho_{AB}^{T_{B}}\leq\left(
\left\lfloor \mu\right\rfloor -1\right)  G_{AB}^{T_{B}},\ G_{AB}%
\geq0,\ \operatorname{Tr}G_{AB}=1,\ \mu\geq 2\right\}  \notag\\
& \leq\inf\left\{  \log_{2}\left\lfloor \mu\right\rfloor :-\left(
\mu-2\right)  G_{AB}^{T_{B}}\leq\rho_{AB}^{T_{B}}\leq\left(  \mu-2\right)
G_{AB}^{T_{B}},\ G_{AB}\geq0,\ \operatorname{Tr}G_{AB}=1,\ \mu\geq 2\right\}
\notag \\
& \leq\inf\left\{  \log_{2}\mu:-\left(  \mu-2\right)  G_{AB}^{T_{B}}\leq
\rho_{AB}^{T_{B}}\leq\left(  \mu-2\right)  G_{AB}^{T_{B}},\ G_{AB}%
\geq0,\ \operatorname{Tr}G_{AB}=1,\ \mu\geq 2\right\}  \notag \\
& =\inf\left\{  \log_{2}(  \operatorname{Tr}[S_{AB}]+2)
:-S_{AB}^{T_{B}}\leq\rho_{AB}^{T_{B}}\leq S_{AB}^{T_{B}},\ S_{AB}%
\geq0\right\}  \notag \\
& =\inf\left\{  \log_{2}\!\left(  2^{E_{\kappa}(\rho_{AB})}+2\right)
:-S_{AB}^{T_{B}}\leq\rho_{AB}^{T_{B}}\leq S_{AB}^{T_{B}},\ S_{AB}%
\geq0\right\}  .
\end{align}
The first equality follows from the definition and the fact that we are focusing on an NPT state, for which it is not possible to have the optimal value be equal to $m=1$. For the first inequality,
consider that the operator inequality $\left(  m-1\right)  G_{AB}^{T_{B}}%
\leq\left(  m+1\right)  G_{AB}^{T_{B}}$ holds so that the following
implication holds%
\begin{equation}
\rho_{AB}^{T_{B}}\leq\left(  m-1\right)  G_{AB}^{T_{B}}\qquad\Longrightarrow
\qquad\rho_{AB}^{T_{B}}\leq\left(  m+1\right)  G_{AB}^{T_{B}}.
\end{equation}
Thus, the optimization in the first line is over a larger set than the
optimization in the second line. The second equality follows by allowing the
optimization to be over all $\mu\in\mathbb{R}$ such that $\mu\geq1$, but then
taking the floor in the constraints on $\mu$ and in the objective function.
The second inequality follows from reasoning similar to that given to justify
the first inequality. In this case, $-(\left\lfloor \mu\right\rfloor
-1)
\leq -(\mu-2)$ and $\mu-2\leq\left\lfloor \mu\right\rfloor
-1$, so that the set over which we are optimizing becomes smaller. The third
inequality follows because $\left\lfloor \mu\right\rfloor \leq\mu$ in the objective function.  The penultimate equality follows because
we can set $S_{AB}=\left(  \mu-2\right)  G_{AB}$, so that $\mu
=\operatorname{Tr}[S_{AB}]+2$ and because the sets $\{\left(  \mu-2\right)
G_{AB}:\mu\geq2,G_{AB}\geq0,\operatorname{Tr}[G_{AB}]=1\}$ and $\{S_{AB}%
\geq0\}$ are the same. The final equality follows from the definition of
$\kappa$-entanglement.
\end{proof}

\begin{remark}
We note that the original proof of the upper bound in \eqref{eq:bounds-for-one-shot-e-kap}, as given in \cite{WW18}, was
incomplete. It was subsequently revised in \cite[Lemma~V.8]{GS19}. For completeness, we have presented
the proof given above, which adopts some ideas from the revision given in \cite[Lemma~V.8]{GS19}.
\end{remark}

 \subsection{Semi-definite programming}\label{app:SDP}
 
The $\kappa$-entanglement $E_{\kappa}(\rho_{AB})$ of a bipartite state $\rho_{AB}$ can be computed by means of a semi-definite program (SDP) \cite{Vandenberghe1996}. Semi-definite optimization is a subfield of convex optimization concerned with the optimization of a linear objective function over the intersection of the cone of positive semi-definite matrices with an affine space. Most interior-point
methods for linear programming have been generalized to SDPs (e.g., \cite{Khachiyan1980}), which have polynomial worst-case complexity and have excellent performance in practice. Note that the CVX software~\cite{Grant2008} allows one to compute SDPs in practice.
 
In the following, we briefly introduce the basics of semi-definite programming, and we base our presentation on \cite{Watrous2011b}.

\begin{definition}
\label{def:SDPs}
	A {semi-definite program (SDP)} is defined by a triplet $\{\Psi,C,D\}$, where $C$ and $D$ are Hermitian operators  and
	$\Psi$ is a Hermiticity-preserving map.
	\vspace{-0.1cm}
	\begin{center}
		\begin{minipage}[t]{1.6in}
			\centerline{\underline{Primal problem}}\vspace{-7mm}
			\begin{align*}
			\inf \quad &  \tr[{D}{Y}]\\
			\text{subject to:}\quad &  \Psi(Y) \ge C,\\
			& Y\ge 0.
			\end{align*}
		\end{minipage}
				\hspace*{10mm}
				\begin{minipage}[t]{1.6in}
			\centerline{\underline{Dual problem}}\vspace{-7mm}
			\begin{align*}
			\sup\quad &  \tr[{C}{X}]\\
			\text{subject to:}\quad &  \Psi^\dag(X) \le D,\\
			& X\ge 0.
			\end{align*}
		\end{minipage}
	\end{center}
	where $\Psi^\dag$ is the dual map to $\Psi$ (i.e., it satisfies $\tr [Y^\dag \Psi^\dag(X)]=\tr [(\Psi(Y))^\dag X]$ for all linear operators $X$ and $Y$). Note that in many cases of interest, the
optimization may not be explicitly written as the standard form above, but one can instead recast it into the above form.
\end{definition}

Weak duality holds for all semi-definite programs, which states that the dual optimum  never exceeds the
primal optimum. The duality theory of SDP is fruitful and useful, and we refer to \cite{Watrous2011b} for more details.

By inspecting its definition, we see that the $\kappa$-entanglement of a bipartite state $\rho_{AB}$ can be computed by solving the following SDP:
\begin{equation}
\label{eq:prime SDP}
\begin{split}
 \inf \ & \tr [S_{AB}] \\
 \text{s.t.}\ & -S_{AB}^{T_B}\le\rho_{AB}^{T_B}\le S_{AB}^{T_B}, \\
 &  \, S_{AB}\ge 0.
\end{split}
\end{equation}
Supposing that the optimal value of the above SDP is $s$, then
 $E_\kappa(\rho_{AB}) = \log_2 s$. 

The Matlab codes for computing $E_\kappa$ are provided  \href{https://github.com/xinwang1/kappa-entanglement}{\textcolor{blue}{online}}.


\section{Properties of $\kappa$-entanglement}
\label{sec:kappa ent}

\subsection{Monotonicity under completely-PPT-preserving channels}

Throughout this work, we consider completely-PPT-preserving operations \cite{R99,R01}, defined as a bipartite operation $\mathcal{P}_{AB\to A'B'}$ (completely positive map) such that the map $T_{B'}\circ \mathcal{P}_{AB\to A'B'} \circ T_B$ is also completely positive, where $T_B$ and $T_{B'}$ denote the partial transpose map acting on the input system $B$ and the output system $B'$, respectively. If $\mathcal{P}_{AB\to A'B'}$ is also trace preserving, such that it is a quantum channel, and $T_{B'}\circ \mathcal{P}_{AB\to A'B'} \circ T_B$ is also completely positive, then we say that $\mathcal{P}_{AB\to A'B'}$ is a completely-PPT-preserving channel.

The most important property of the $\kappa$-entanglement measure is that it does not increase under the action of a completely-PPT-preserving channel. In fact, we prove a stronger statement, that $\kappa$-entanglement does not increase under selective completely-PPT-preserving operations.  Note that an LOCC channel \cite{Bennett1996c,CLMOW14}, as considered in entanglement theory, is a special kind of 
completely-PPT-preserving channel, as observed in \cite{R99,R01}. 

\begin{theorem}
[Monotonicity]
\label{prop:ent-monotone}Let $\rho_{AB}$ be a quantum state acting on a separable Hilbert
space, and let $\{\mathcal{P}_{AB\rightarrow A^{\prime}B^{\prime}}^{x}\}_{x}$
be a set of completely positive, trace non-increasing maps that are each completely
PPT-preserving, such that the sum map $\sum_{x}\mathcal{P}_{AB\rightarrow
A^{\prime}B^{\prime}}^{x}$ is quantum channel. Then the following entanglement
monotonicity inequality holds%
\begin{equation}
E_{\kappa}(\rho_{AB})\geq\sum_{x \, : \, p(x) > 0}p(x)E_{\kappa}\!\left(  \frac{\mathcal{P}%
_{AB\rightarrow A^{\prime}B^{\prime}}^{x}(\rho_{AB})}{p(x)}\right)
,\label{eq:mono-selective-E-kappa}%
\end{equation}
where $p(x)\coloneqq\operatorname{Tr}[\mathcal{P}_{AB\rightarrow A^{\prime}B^{\prime}%
}^{x}(\rho_{AB})]$. In particular, for a
completely-PPT-preserving quantum channel $\mathcal{P}_{AB\rightarrow
A^{\prime}B^{\prime}}$, the following inequality holds%
\begin{equation}
E_{\kappa}(\rho_{AB})\geq E_{\kappa}\!\left(  \mathcal{P}_{AB\rightarrow
A^{\prime}B^{\prime}}(\rho_{AB})\right)
.\label{eq:mono-non-selective-e-kappa}%
\end{equation}

\end{theorem}

\begin{proof}
Let $S_{AB}$ be such that%
\begin{equation}
S_{AB}\geq0, \qquad -S_{AB}^{T_{B}}\leq\rho_{AB}^{T_{B}}\leq S_{AB}^{T_{B}}%
.\label{eq:feasible-condition-for-e-kappa}%
\end{equation}
Since $\mathcal{P}_{AB\rightarrow A^{\prime}B^{\prime}}^{x}$ is
completely-PPT-preserving, we have that%
\begin{equation}
-(T_{B^{\prime}}\circ\mathcal{P}_{AB\rightarrow A^{\prime}B^{\prime}}^{x}\circ
T_{B})(S_{AB}^{T_{B}})\leq(T_{B^{\prime}}\circ\mathcal{P}_{AB\rightarrow
A^{\prime}B^{\prime}}^{x}\circ T_{B})(\rho_{AB}^{T_{B}})\leq(T_{B^{\prime}%
}\circ\mathcal{P}_{AB\rightarrow A^{\prime}B^{\prime}}^{x}\circ T_{B}%
)(S_{AB}^{T_{B}}),
\end{equation}
which reduces to the following for all $x$ such that $p(x)>0$:%
\begin{equation}
-\frac{[\mathcal{P}_{AB\rightarrow A^{\prime}B^{\prime}}^{x}(S_{AB}%
)]^{T_{B^{\prime}}}}{p(x)}\leq\frac{\lbrack\mathcal{P}_{AB\rightarrow
A^{\prime}B^{\prime}}^{x}(\rho_{AB})]^{T_{B^{\prime}}}}{p(x)}\leq\frac
{\lbrack\mathcal{P}_{AB\rightarrow A^{\prime}B^{\prime}}^{x}(S_{AB}%
)]^{T_{B^{\prime}}}}{p(x)}.\label{eq:feasible-for-mono-1}%
\end{equation}
Furthermore, since $S_{AB}\geq0$ and $\mathcal{P}_{AB\rightarrow A^{\prime
}B^{\prime}}^{x}$ is completely positive, we conclude the following for all $x$ such that $p(x)>0$:
\begin{equation}
\frac{\mathcal{P}_{AB\rightarrow A^{\prime}B^{\prime}}^{x}(S_{AB})}{p(x)}%
\geq0.\label{eq:feasible-for-mono-2}%
\end{equation}
Thus, the operator $\frac{\mathcal{P}_{AB\rightarrow A^{\prime}B^{\prime}}%
^{x}(S_{AB})}{p(x)}$ is feasible for $E_{\kappa}\!\left(  \frac{\mathcal{P}%
_{AB\rightarrow A^{\prime}B^{\prime}}^{x}(\rho_{AB})}{p(x)}\right)  $. Then we
find that%
\begin{align}
\log_2\operatorname{Tr}[S_{AB}]  & =\log_2\sum_{x}\operatorname{Tr}\mathcal{P}%
_{AB\rightarrow A^{\prime}B^{\prime}}^{x}(S_{AB})\\
& =\log_2\sum_{x \, : \, p(x)>0}p(x)\operatorname{Tr}\!\left[\frac{\mathcal{P}_{AB\rightarrow
A^{\prime}B^{\prime}}^{x}(S_{AB})}{p(x)}\right]\\
& \geq\sum_{x  \, : \, p(x)>0}p(x)\log_2\operatorname{Tr}\!\left[\frac{\mathcal{P}_{AB\rightarrow
A^{\prime}B^{\prime}}^{x}(S_{AB})}{p(x)}\right]\\
& \geq\sum_{x  \, : \, p(x)>0}p(x)E_{\kappa}\!\left(  \frac{\mathcal{P}_{AB\rightarrow
A^{\prime}B^{\prime}}^{x}(\rho_{AB})}{p(x)}\right)  .
\end{align}
The first equality follows from the assumption that the sum map $\sum
_{x}\mathcal{P}_{AB\rightarrow A^{\prime}B^{\prime}}^{x}$ is trace preserving.
The first inequality follows from concavity of the logarithm. The second
inequality follows from the definition of $E_{\kappa}$ and the fact that
$\frac{\mathcal{P}_{AB\rightarrow A^{\prime}B^{\prime}}^{x}(S_{AB})}{p(x)}$
satisfies \eqref{eq:feasible-for-mono-1} and \eqref{eq:feasible-for-mono-2}.
Since the inequality holds for an arbitrary $S_{AB}\geq0$ satisfying
\eqref{eq:feasible-condition-for-e-kappa}, we conclude the inequality in~\eqref{eq:mono-selective-E-kappa}. 

The inequality in
\eqref{eq:mono-non-selective-e-kappa}\ is a special case of that in
\eqref{eq:mono-selective-E-kappa}, in which the set $\{\mathcal{P}%
_{AB\rightarrow A^{\prime}B^{\prime}}^{x}\}_{x}$ is a singleton, consisting of a single
completely-PPT-preserving quantum channel.
\end{proof}

\subsection{Dual representation}

The optimization problem dual to $E_\kappa(\rho_{AB})$ in Definition~\ref{def:kappa-ent} is as follows:
\begin{equation}
\label{eq:a dual}
E^{\text{dual}}_\kappa(\rho_{AB}) \coloneqq\sup \{ \log_2 \tr \rho_{AB}(V_{AB}-W_{AB}):
 V_{AB}+W_{AB}\le \1_{AB},\, V_{AB}^{T_B},\, W_{AB}^{T_B}\ge 0\}.
\end{equation}
We show a proof of this in Section~\ref{sec:dual-derive}.
By weak duality \cite[Section~1.2.2]{Watrous2011b}, we have for any bipartite state $\rho_{AB}$ acting on a separable Hilbert space that
\begin{equation}
E^{\text{dual}}_\kappa(\rho_{AB})\leq E_\kappa(\rho_{AB}) .
\label{eq:weak-dual-kappa}
\end{equation}
For all finite-dimensional states $\rho_{AB}$, strong duality holds, so that
\begin{equation}
E_\kappa(\rho_{AB}) = E^{\text{dual}}_\kappa(\rho_{AB}).
\label{eq:strong-dual-e-kappa}
\end{equation}
This follows as a consequence of Slater's theorem and by choosing $V_{AB} =  \1_{AB}/2$ and $W_{AB} = \1_{AB} / 3$.

By employing the strong duality equality in \eqref{eq:strong-dual-e-kappa} for the finite-dimensional case, along with the approach from \cite{FAR11}, we conclude that the following equality holds for all bipartite states $\rho_{AB}$ acting on a separable Hilbert space:
\begin{equation}
E_\kappa(\rho_{AB}) = E^{\text{dual}}_\kappa(\rho_{AB}).
\label{eq:strong-dual-e-kappa-inf-dim}
\end{equation}
We provide an explicit proof of \eqref{eq:strong-dual-e-kappa-inf-dim} in Section~\ref{app:kappa-to-dual-infty}.

\subsubsection{Derivation of dual}

\label{sec:dual-derive}

Let us derive the dual form $\kappa$-entanglement in more detail. First we can
rewrite the primal SDP\ for the $\kappa$-entanglement in the standard form in Definition~\ref{def:SDPs} as%
\begin{equation}
\inf_{Y\geq0}\left\{  \operatorname{Tr}[DY]:\Psi(Y)\geq C\right\}  ,
\end{equation}
with%
\begin{equation}
D=I,\quad Y=S_{AB},\quad\Psi(Y)=%
\begin{bmatrix}
S_{AB}^{T_{B}} & 0\\
0 & S_{AB}^{T_{B}}%
\end{bmatrix}
,\quad C=%
\begin{bmatrix}
\rho_{AB}^{T_{B}} & 0\\
0 & -\rho_{AB}^{T_{B}}%
\end{bmatrix}
.
\end{equation}
Then setting%
\begin{equation}
X=%
\begin{bmatrix}
V_{AB} & 0\\
0 & W_{AB}%
\end{bmatrix}
,
\end{equation}
with $V_{AB},W_{AB}\geq0$, we find that%
\begin{align}
\operatorname{Tr}[\Psi(Y)X]  & =\operatorname{Tr}\left[
\begin{bmatrix}
S_{AB}^{T_{B}} & 0\\
0 & S_{AB}^{T_{B}}%
\end{bmatrix}%
\begin{bmatrix}
V_{AB} & 0\\
0 & W_{AB}%
\end{bmatrix}
\right]  \\
& =\operatorname{Tr}[S_{AB}^{T_{B}}(V_{AB}+W_{AB})]\\
& =\operatorname{Tr}[S_{AB}(V_{AB}+W_{AB})^{T_{B}}],
\end{align}
implying that%
\begin{equation}
\Psi^{\dag}(X)=(V_{AB}+W_{AB})^{T_{B}}.
\end{equation}
Then plugging into the standard form in Definition~\ref{def:SDPs}, we find that the dual is given by%
\begin{align}
& \sup_{X\geq0}\left\{  \operatorname{Tr}[CX]:\Psi^{\dag}(X)\leq D\right\}
\nonumber\\
& =\sup_{V_{AB},W_{AB}\geq0}\left\{  \operatorname{Tr}[\rho_{AB}^{T_{B}%
}\left(  V_{AB}-W_{AB}\right)  ]:(V_{AB}+W_{AB})^{T_{B}}\leq I_{AB}\right\}
\\
& =\sup_{V_{AB},W_{AB}\geq0}\left\{  \operatorname{Tr}[\rho_{AB}\left(
V_{AB}-W_{AB}\right)  ^{T_{B}}]:(V_{AB}+W_{AB})^{T_{B}}\leq I_{AB}\right\}
\label{eq:dual-for-witnessed-ent}\\
& =\sup_{V_{AB}^{T_{B}},W_{AB}^{T_{B}}\geq0}\left\{  \operatorname{Tr}%
[\rho_{AB}\left(  V_{AB}-W_{AB}\right)  ]:V_{AB}+W_{AB}\leq I_{AB}\right\}  .
\end{align}
The last line follows from the substitutions $V_{AB}\rightarrow V_{AB}^{T_{B}%
}$ and $W_{AB}\rightarrow W_{AB}^{T_{B}}$. Thus,%
\begin{equation}
E_{\kappa}^{\text{dual}}(\rho_{AB})=\log_{2}\sup_{V_{AB}^{T_{B}},W_{AB}%
^{T_{B}}\geq0}\left\{  \operatorname{Tr}[\rho_{AB}\left(  V_{AB}%
-W_{AB}\right)  ]:V_{AB}+W_{AB}\leq I_{AB}\right\}  ,
\end{equation}
as claimed.

\subsection{$\kappa$-entanglement as witnessed entanglement}

We note that the dual form of the $\kappa$-entanglement allows for
understanding it in terms of witnessed entanglement, the latter being a
concept discussed in \cite{FB05}. Recall that an entanglement witness is a Hermitian
operator $W_{AB}$ that satisfies $\operatorname{Tr}[W_{AB}\rho_{AB}]<0$ for an
entangled state $\rho_{AB}$ and $\operatorname{Tr}[W_{AB}\sigma_{AB}]\geq0$
for all separable states. The idea behind witnessed entanglement is to
intersect the set of all witnesses with another set $\mathcal{C}$ (call the
intersection $\mathcal{M}$) and then perform the following optimization:%
\begin{equation}
\max\left\{  0,-\inf_{W_{AB}\in\mathcal{M}}\operatorname{Tr}[W_{AB}\rho
_{AB}]\right\}  .\label{eq:witnessed-ent-gen-def}%
\end{equation}
As discussed in \cite{FB05}, this approach allows for quantifying entanglement.

We can arrive at the conclusion that $\kappa$-entanglement is a particular
kind of witnessed entanglement by exploiting the dual form and the equality
presented in \eqref{eq:strong-dual-e-kappa-inf-dim} and more generally in Section~\ref{app:kappa-to-dual-infty}. Recalling the dual form given in
\eqref{eq:dual-for-witnessed-ent}, we can write it as follows:%
\begin{align}
& E_{\kappa}^{\text{dual}}(\rho_{AB})  \notag \\
& =\log_{2}\sup\left\{  \operatorname{Tr}%
[\left(  K_{AB}-L_{AB}\right)  ^{T_{B}}\rho_{AB}]:\left(  L_{AB}%
+K_{AB}\right)  ^{T_{B}}\leq I_{AB},\ K_{AB},L_{AB}\geq0\right\}  \\
& =\log_{2}\sup\left\{  -\operatorname{Tr}[\left(  L_{AB}-K_{AB}\right)
^{T_{B}}\rho_{AB}]:\left(  L_{AB}+K_{AB}\right)  ^{T_{B}}\leq I_{AB}%
,\ K_{AB},L_{AB}\geq0\right\}  .\\
& =\log_{2}\left[  -\inf\left\{  \operatorname{Tr}[\left(  L_{AB}%
-K_{AB}\right)  ^{T_{B}}\rho_{AB}]:\left(  L_{AB}+K_{AB}\right)  ^{T_{B}}\leq
I_{AB},\ K_{AB},L_{AB}\geq0\right\}  \right]  .
\end{align}
Now setting $Z_{AB}$ and $Y_{AB}$ to be the following Hermitian operators:%
\begin{equation}
Z_{AB}:=L_{AB}-K_{AB},\qquad Y_{AB}:=L_{AB}+K_{AB},
\end{equation}
we find that%
\begin{align}
\operatorname{Tr}[\left(  L_{AB}-K_{AB}\right)  ^{T_{B}}\rho_{AB}] &
=\operatorname{Tr}[Z_{AB}^{T_{B}}\rho_{AB}],\\
\left(  L_{AB}+K_{AB}\right)  ^{T_{B}}\leq I_{AB}\quad &  \Longleftrightarrow
\quad Y_{AB}^{T_{B}}\leq I_{AB},\\
K_{AB},L_{AB}\geq0\quad &  \Longleftrightarrow\quad Y_{AB}-Z_{AB}%
\geq0,\ Z_{AB}+Y_{AB}\geq0.
\end{align}
Then we can rewrite the dual as%
\begin{equation}
E_{\kappa}^{\text{dual}}(\rho_{AB})=\log_{2}\left[  -\inf_{Z_{AB},Y_{AB}%
\in\text{Herm}}\left\{  \operatorname{Tr}[Z_{AB}^{T_{B}}\rho_{AB}%
]:Y_{AB}^{T_{B}}\leq I_{AB},\ -Y_{AB}\leq Z_{AB}\leq Y_{AB}\right\}  \right]
.
\end{equation}
We can then perform the final substitutions $Z_{AB}^{T_{B}}\rightarrow Z_{AB}$
and $Y_{AB}^{T_{B}}\rightarrow Y_{AB}$ in order to be fully consistent with
the formulation in \eqref{eq:witnessed-ent-gen-def}. So then%
\begin{equation}
E_{\kappa}^{\text{dual}}(\rho_{AB})=\log_{2}\left[  -\inf_{Z_{AB},Y_{AB}%
\in\text{Herm}}\left\{  \operatorname{Tr}[Z_{AB}\rho_{AB}]:Y_{AB}\leq
I_{AB},\ -Y_{AB}^{T_{B}}\leq Z_{AB}^{T_{B}}\leq Y_{AB}^{T_{B}}\right\}
\right]  .
\end{equation}
Defining the set%
\begin{equation}
\mathcal{K}:=\left\{  Z_{AB}:Z_{AB}\in\text{Herm},\ \exists Y_{AB}%
\in\text{Herm},\ Y_{AB}\leq I_{AB},\ -Y_{AB}^{T_{B}}\leq Z_{AB}^{T_{B}}\leq
Y_{AB}^{T_{B}}\right\}  ,
\end{equation}
we can write $E_{\kappa}^{\text{dual}}(\rho_{AB})$ as follows:%
\begin{equation}
E_{\kappa}^{\text{dual}}(\rho_{AB})=\log_{2}\left[  -\inf_{Z_{AB}%
\in\mathcal{K}}\operatorname{Tr}[Z_{AB}\rho_{AB}]\right]
\label{eq:final-form-kappa-witnessed-ent}
\end{equation}
Note that%
\begin{equation}
\inf_{Z_{AB}\in\mathcal{K}}\operatorname{Tr}[Z_{AB}\rho_{AB}]\leq-1
\end{equation}
because we can always pick $Z_{AB}=-I_{AB}$ and $Y_{AB}=I_{AB}$, and we get
that $\operatorname{Tr}[Z_{AB}\rho_{AB}]=-1$ for these choices. So this
implies that the expression inside the logarithm in \eqref{eq:final-form-kappa-witnessed-ent} is never smaller than one.

Any $Z_{AB}$ chosen from the set $\mathcal{K}$ leads to an entanglement
witness after summing it with the identity. In this case, we have
$\operatorname{Tr}[\left(  Z_{AB}+I_{AB}\right)  \rho_{AB}]\geq0$ for
PPT\ states. To see this, suppose that $\rho_{AB}$ is separable and thus PPT.
Then we find for any $Z_{AB}\in\mathcal{K}$ that%
\begin{align}
\operatorname{Tr}[Z_{AB}\rho_{AB}] &  =\operatorname{Tr}[Z_{AB}^{T_{B}}%
\rho_{AB}^{T_{B}}]\\
&  \geq-\operatorname{Tr}[Y_{AB}^{T_{B}}\rho_{AB}^{T_{B}}]\\
&  =-\operatorname{Tr}[Y_{AB}\rho_{AB}]\\
&  \geq-\operatorname{Tr}[I_{AB}\rho_{AB}]\\
&  =-1,
\end{align}
which implies that%
\begin{equation}
\operatorname{Tr}[\left(  Z_{AB}+I_{AB}\right)  \rho_{AB}]\geq0,
\end{equation}
for all PPT\ states. The first inequality follows from the assumption that
$\rho_{AB}^{T_{B}}\geq0$ and the last from the assumption that $\rho_{AB}%
\geq0$. Thus, every $Z_{AB}\in\mathcal{K}$ summed with the identity is an
entanglement witness. Additionally, by employing the faithfulness of $\kappa
$-entanglement (discussed later in Proposition~\ref{prop:faithfulness}), we conclude that for every
NPT\ state, there exists $Z_{AB}\in\mathcal{K}$ such that $\operatorname{Tr}%
[\left(  Z_{AB}+I_{AB}\right)  \rho_{AB}]<0$. So the development above
establishes a link between $\kappa$-entanglement and the theory of
entanglement witnesses.

\subsection{Additivity}

Both the primal and dual SDPs for $E_\kappa$ are important, as the combination of them allows for proving the following additivity of $E_\kappa$ with respect to tensor-product states.

\begin{proposition}[Additivity]\label{lemma:add}
	For any two bipartite states $\rho_{AB}$ and $\omega_{A'B'}$ acting on separable Hilbert spaces, the following additivity identity holds
	\begin{equation}
	E_\kappa(\rho_{AB}\ox\omega_{A'B'})
	=E_\kappa(\rho_{AB})+E_\kappa(\omega_{A'B'}).
	\label{eq:additivity-e-kappa-states}
	\end{equation}
	where the bipartition on the left-hand side is understood as $AA'|BB'$.
\end{proposition}

\begin{proof}
	From Definition~\ref{def:kappa-ent}, we can write $E_\kappa(\rho_{AB})$ as 
	\begin{align}
	E_\kappa(\rho_{AB})=\inf \{\log_2 \tr S_{AB}: -S^{T_B}_{AB} \le \rho_{AB}^{T_B} \le S^{T_B}_{AB}, \, S_{AB}\ge 0\}.
		\end{align}
		Let $S_{AB}$ be an arbitrary operator satisfying $-S^{T_B}_{AB} \le \rho_{AB}^{T_B} \le S^{T_B}_{AB}, \, S_{AB}\ge 0$, and let  $R_{A'B'}$ be an arbitrary operator satisfying $-R^{T_{B'}}_{A'B'} \le \omega_{A'B'}^{T_{B'}} \le R^{T_{B'}}_{A'B'}, \, R_{A'B'}\ge 0$.
		From these inequalities, we conclude that
		\begin{align}
	 0 & \leq ( S^{T_B}_{AB} + \rho_{AB}^{T_B}) \otimes ( R^{T_{B'}}_{A'B'} + \omega_{A'B'} ^{T_{B'}}) \notag \\
	 & =
	S^{T_B}_{AB} \otimes R^{T_{B'}}_{A'B'} + \rho^{T_B}_{AB} \otimes R^{T_{B'}}_{A'B'} 
	+S^{T_B}_{AB} \otimes \omega^{T_{B'}}_{A'B'} 
	+\rho^{T_B}_{AB} \otimes \omega^{T_{B'}}_{A'B'} \label{eq:for-additivity-1} ,
	  \\
	0 & \leq (S^{T_B}_{AB} - \rho_{AB}^{T_B} 	) \otimes (R^{T_{B'}}_{A'B'} - \omega_{A'B'} ^{T_{B'}} )\notag \\
	& = S^{T_B}_{AB} \otimes R^{T_{B'}}_{A'B'} - \rho^{T_B}_{AB} \otimes R^{T_{B'}}_{A'B'} 
	-S^{T_B}_{AB} \otimes \omega^{T_{B'}}_{A'B'} 
	+\rho^{T_B}_{AB} \otimes \omega^{T_{B'}}_{A'B'} \label{eq:for-additivity-2} .
		\end{align}
		Adding \eqref{eq:for-additivity-1} and \eqref{eq:for-additivity-2}, we conclude that
		\begin{equation}
				-(S_{AB} \otimes R_{A'B'})^{T_{BB'}} \le (\rho_{AB} \otimes \omega_{A'B'})^{T_{BB'}}.
		\end{equation}
		Furthermore, we conclude that
		\begin{align}
	 0 & \leq ( S^{T_B}_{AB} + \rho_{AB}^{T_B}) \otimes ( R^{T_{B'}}_{A'B'} - \omega_{A'B'} ^{T_{B'}}) \notag \\
	 & =
	S^{T_B}_{AB} \otimes R^{T_{B'}}_{A'B'} + \rho^{T_B}_{AB} \otimes R^{T_{B'}}_{A'B'} 
	-S^{T_B}_{AB} \otimes \omega^{T_{B'}}_{A'B'} 
	-\rho^{T_B}_{AB} \otimes \omega^{T_{B'}}_{A'B'} \label{eq:for-additivity-3} ,
	  \\
	0 & \leq (S^{T_B}_{AB} - \rho_{AB}^{T_B} 	) \otimes (R^{T_{B'}}_{A'B'} + \omega_{A'B'} ^{T_{B'}} )\notag \\
	& = S^{T_B}_{AB} \otimes R^{T_{B'}}_{A'B'} - \rho^{T_B}_{AB} \otimes R^{T_{B'}}_{A'B'} 
	+S^{T_B}_{AB} \otimes \omega^{T_{B'}}_{A'B'} 
	-\rho^{T_B}_{AB} \otimes \omega^{T_{B'}}_{A'B'} \label{eq:for-additivity-4} .
		\end{align}
		Adding \eqref{eq:for-additivity-3} and \eqref{eq:for-additivity-4}, we conclude that
		\begin{equation}
				(\rho_{AB} \otimes \omega_{A'B'})^{T_{BB'}} \le (S_{AB} \otimes R_{A'B'})^{T_{BB'}}.
		\end{equation}
		Then it follows that
		\begin{equation}
		-(S_{AB} \otimes R_{A'B'})^{T_{BB'}} \le (\rho_{AB} \otimes \omega_{A'B'})^{T_{BB'}} \le (S_{AB} \otimes R_{A'B'})^{T_{BB'}}, \, S_{AB} \otimes R_{A'B'}\ge 0,
		\label{eq:for-additivity-5}
		\end{equation}
		so that
		\begin{equation}
		E_\kappa(\rho_{AB}\ox\omega_{A'B'}) \leq \log_2 \tr S_{AB} \otimes R_{A'B'} = \log_2 \tr S_{AB} +\log_2 \tr R_{A'B'}.
		\end{equation}
		Since the inequality holds for all $S_{AB}$ and $R_{A'B'}$ satisfying the constraints above, we conclude that
	\begin{align}\label{eq:a sub}
		E_\kappa(\rho_{AB}\ox\omega_{A'B'})\le E_\kappa(\rho_{AB})+E_\kappa(\omega_{A'B'}).
	\end{align}
	
	To see the super-additivity of $E_\kappa$, i.e., the opposite inequality, let  $\{V^1_{AB},W^1_{AB}\}$ and $\{V^2_{A'B'},W^2_{A'B'}\}$
	be arbitrary operators satisfying the conditions in \eqref{eq:a dual} for
	$\rho_{AB}$ and $\omega_{A'B'}$, respectively.
	Now we choose 
	\begin{align}
		R_{ABA'B'}
		&=V^1_{AB} \ox V^2_{A'B'}
		+ W^1_{AB} \ox W^2_{A'B'},\\
		S_{ABA'B'}
		&=V^1_{AB} \ox W^2_{A'B'}
		+ W^1_{AB} \ox V^2_{A'B'}.
	\end{align}
	One can verify from \eqref{eq:a dual} that
	\begin{align}
	R_{ABA'B'}^{T_{BB'}},\, S_{ABA'B'}^{T_{BB'}}& \ge 0 , \\
	R_{ABA'B'}+S_{ABA'B'} & =(V^1_{AB}+W^1_{AB})\ox(V^2_{AB}+W^2_{AB})\le \1_{ABA'B'},
	\end{align}
	which implies that $\{R_{ABA'B'},S_{ABA'B'}\}$ is a feasible solution to  \eqref{eq:a dual} for $E_\kappa(\rho_{AB}\ox\omega_{A'B'})$. 
		Thus, we have that
	\begin{align}	 
	E^{\text{dual}}_\kappa(\rho_{AB}\ox\omega_{A'B'})
	&
	\ge \log_2 \tr (\rho_{AB}\ox\omega_{A'B'}) (R_{ABA'B'}-S_{ABA'B'})\\
	&=\log_2 [\tr \rho_{AB} (V^1_{AB}-W^1_{AB}) \cdot \tr\omega_{A'B'} (V^2_{A'B'}-W^2_{A'B'})]\\
	&= \log_2 (\tr \rho_{AB} (V^1_{AB}-W^1_{AB})) + \log_2( \tr\omega_{A'B'} (V^2_{A'B'}-W^2_{A'B'})).
	\end{align}
	Since the inequality has been shown for arbitrary $\{V^1_{AB},W^1_{AB}\}$ and $\{V^2_{A'B'},W^2_{A'B'}\}$
	satisfying the conditions in \eqref{eq:a dual} for
	$\rho_{AB}$ and $\omega_{A'B'}$, respectively, we conclude that
	\begin{equation}
	\label{eq:super-add-kappa}
	E^{\text{dual}}_\kappa(\rho_{AB}\ox\omega_{A'B'}) \geq 
E^{\text{dual}}_\kappa(\rho_{AB})+E^{\text{dual}}_\kappa(\omega_{A'B'}).
	\end{equation}
	Applying \eqref{eq:a sub}, \eqref{eq:super-add-kappa}, \eqref{eq:strong-dual-e-kappa-inf-dim}, and the more general equality in Section~\ref{app:kappa-to-dual-infty}, we conclude \eqref{eq:additivity-e-kappa-states}.
\end{proof}


\subsection{Relation to logarithmic negativity}

The following proposition establishes an inequality  relating $E_\kappa$ to the logarithmic negativity \cite{Vidal2002,Plenio2005b}, defined as in \eqref{eq:log-neg}.

\begin{proposition}
\label{prop:connect-to-log-neg}
Let $\rho_{AB}$ be a  bipartite state acting on a separable Hilbert space.
Then%
\begin{equation}
E_{\kappa}(\rho_{AB})\geq E_{N}(\rho_{AB}%
).\label{eq:kappa-greater-than-log-neg}%
\end{equation}
If $\rho_{AB}$ satisfies the condition $|\rho_{AB}^{T_B}|^{T_B} \geq 0$, then
\begin{equation}
E_{\kappa}(\rho_{AB})= E_{N}(\rho_{AB}%
).\label{eq:kappa-equal-log-neg}%
\end{equation}
\end{proposition}

\begin{proof}
Consider from the dual formulation of $E_{\kappa}(\rho_{AB})$ in \eqref{eq:a dual} that%
\begin{equation}
E^{\text{dual}}_{\kappa}(\rho_{AB})  =\sup\log_2 \{ \operatorname{Tr}\rho_{AB}(V_{AB}-W_{AB}):  V_{AB}+W_{AB}\leq \1_{AB},\ \ V_{AB}^{T_{B}},\ W_{AB}^{T_{B}%
}\geq0 \}.
\end{equation}
Using the fact that the transpose map is its own adjoint, we have that%
\begin{equation}
E^{\text{dual}}_{\kappa}(\rho_{AB})  =\sup\log_2\{\operatorname{Tr}\rho_{AB}^{T_{B}}%
(V_{AB}^{T_{B}}-W_{AB}^{T_{B}})
 :  V_{AB}+W_{AB}\leq \1_{AB},\ \ V_{AB}^{T_{B}},\ W_{AB}^{T_{B}%
}\geq 0 \} .
\end{equation}
Then by a substitution, we can write this as%
\begin{equation}
E^{\text{dual}}_{\kappa}(\rho_{AB})   =\sup\log_2 \{ \operatorname{Tr}\rho_{AB}^{T_{B}}%
(V_{AB}-W_{AB}) :
 V_{AB}^{T_{B}}+W_{AB}^{T_{B}}\leq \1_{AB},\ \ V_{AB}%
,\ W_{AB}\geq0 \} .
\label{eq:dual-for-log-neg}
\end{equation}
Consider a decomposition of $\rho_{AB}^{T_{B}}$ into its positive and negative
part%
\begin{equation}
\rho_{AB}^{T_{B}}=P_{AB}-N_{AB}.
\end{equation}
Let $\Pi_{AB}^{P}$ be the projection onto the positive part, and let $\Pi
_{AB}^{N}$ be the projection onto the negative part. Consider that%
\begin{equation}
\left\vert \rho_{AB}^{T_{B}}\right\vert =P_{AB}+N_{AB}.
\end{equation}
Then we can pick $V_{AB}=\Pi_{AB}^{P}\geq0$ and $W_{AB}=\Pi_{AB}^{N}\geq0$ in \eqref{eq:dual-for-log-neg}, to
find that%
\begin{align}
\operatorname{Tr}\rho_{AB}^{T_{B}}(\Pi_{AB}^{P}-\Pi_{AB}^{N})  &
=\operatorname{Tr}\left(  P_{AB}-N_{AB}\right)  (\Pi_{AB}^{P}-\Pi_{AB}^{N})\\
& =\operatorname{Tr}P_{AB}\Pi_{AB}^{P}+N_{AB}\Pi_{AB}^{N}\\
& =\operatorname{Tr}P_{AB}+N_{AB}\\
& =\operatorname{Tr}\left\vert \rho_{AB}^{T_{B}}\right\vert 
 =\left\Vert \rho_{AB}^{T_{B}}\right\Vert _{1}.
\end{align}
Furthermore, we have for this choice that%
\begin{align}
V_{AB}^{T_{B}}+W_{AB}^{T_{B}}  & =\left(  \Pi_{AB}^{P}\right)  ^{T_{B}%
}+\left(  \Pi_{AB}^{N}\right)  ^{T_{B}}\\
& =\left(  \Pi_{AB}^{P}+\Pi_{AB}^{N}\right)  ^{T_{B}}
 =\1_{AB}^{T_{B}}
 =\1_{AB}.
\end{align}
So this implies the inequality in \eqref{eq:kappa-greater-than-log-neg}, after combining with \eqref{eq:weak-dual-kappa}.

If $\rho_{AB}$ satisfies the condition $|\rho_{AB}^{T_B}|^{T_B} \geq 0$, then we pick $S_{AB} = |\rho_{AB}^{T_B}|$ in \eqref{eq:a prime} and conclude that
\begin{equation}
E_\kappa(\rho_{AB}) \leq E_N(\rho_{AB}).
\end{equation}
Combining with \eqref{eq:kappa-greater-than-log-neg} gives \eqref{eq:kappa-equal-log-neg} for this special case.
\end{proof}

\subsection{Normalization, faithfulness, no convexity, no monogamy}

In this section, we prove that $E_\kappa$ is normalized on maximally entangled states, and for finite-dimensional states, that it achieves its largest value on maximally entangled states. We also show that $E_\kappa$ is faithful, in the sense that it is non-negative and equal to zero if and only if the state is a PPT state. Finally, we provide simple examples that demonstrate that $E_\kappa$ is neither convex nor monogamous.

\begin{proposition}
[Normalization]\label{prop:normalize MES}Let $\Phi_{AB}^{M}$ be a maximally entangled state of Schmidt
rank~$M$. Then%
\begin{equation}\label{eq:normalize MES}
E_{\kappa}(\Phi_{AB}^{M})=\log_2 M.
\end{equation}
Furthermore, for any bipartite state $\rho_{AB}$, the following bound holds
\begin{equation}\label{eq:dim bound of EK}
E_{\kappa}(\rho_{AB})\le \log_2 \min\{d_A,  d_B\},
\end{equation}
where $d_A$ and $d_B$ denote the dimensions of systems $A$ and $B$, respectively.
\end{proposition}

\begin{proof}
Consider that $\Phi_{AB}^{M}$ satisfies the condition 
$\left\vert (\Phi_{AB}^{M})^{T_{B}}\right\vert ^{T_{B}}\geq 0$
because
\begin{equation}
\left\vert (\Phi_{AB}^{M})^{T_{B}}\right\vert ^{T_{B}}=\frac{1}{M}\left\vert
F_{AB}\right\vert ^{T_{B}}=\frac{1}{M}\left(  \1_{AB}\right)  ^{T_{B}}=\frac
{1}{M}\1_{AB}\geq0,
\end{equation}
where $F_{AB}$ is the unitary swap operator, such that $F_{AB}=\Pi
_{AB}^{\mathcal{S}}-\Pi_{AB}^{\mathcal{A}}$, with $\Pi_{AB}^{\mathcal{S}}$ and
$\Pi_{AB}^{\mathcal{A}}$ the respective projectors onto the symmetric and
antisymmetric subspaces. Thus, by Proposition~\ref{prop:connect-to-log-neg}, it follows that%
\begin{align}
E_{\kappa}(\Phi_{AB}^{M})  & =E_{N}(\Phi_{AB}^{M})=\log_2\left\Vert (\Phi
_{AB}^{M})^{T_{B}}\right\Vert _{1}=\log_2\operatorname{Tr}\left\vert (\Phi
_{AB}^{M})^{T_{B}}\right\vert \\
& =\log_2\operatorname{Tr}\frac{1}{M}\left\vert F_{AB}\right\vert =\log_2 \frac{1}%
{M} \operatorname{Tr}\1_{AB}=\log_2 M,
\end{align}
demonstrating \eqref{eq:normalize MES}.

To see \eqref{eq:dim bound of EK}, let us suppose without loss of generality that $d_A\le d_B$.
Given the bipartite state $\rho_{AB}$, Bob can first locally prepare a state $\rho_{AB}$ and teleport the $A$ system to Alice using a maximally entangled state $\Phi^{d_A}$ shared with Alice, which implies that there exists a completely-PPT-preserving channel that converts $\Phi^{d_A}$ to $\rho_{AB}$. Therefore, by the monotonicity of $E_\kappa$ with respect to completely-PPT-preserving channels (Theorem~\ref{prop:ent-monotone}), we find that
\begin{equation}
\log_2 d_A=\EK(\Phi^{d_A}) \ge \EK(\rho_{AB}).
\end{equation}
This concludes the proof.
\end{proof}

\begin{proposition}
[Faithfulness]
\label{prop:faithfulness}For a state $\rho_{AB}$ acting on a separable Hilbert space, we
have that $E_{\kappa}(\rho_{AB})\geq0$ and $E_{\kappa}(\rho_{AB})=0$ if and
only if $\rho_{AB}^{T_{B}}\geq 0$.
\end{proposition}

\begin{proof}
To see that $E_{\kappa}(\rho_{AB})\geq0$, take $V_{AB}=\1_{AB}$ and $W_{AB}=0$
in \eqref{eq:a dual}, so that $E_{\kappa}^{\text{dual}}(\rho_{AB})\geq0$. Then we conclude that
$E_{\kappa}(\rho_{AB})\geq0$ from the weak duality inequality in \eqref{eq:weak-dual-kappa}.

Now suppose that $\rho_{AB}^{T_{B}}\geq0$. Then we can set $S_{AB}=\rho_{AB}$ in \eqref{eq:a prime},
so that the conditions $-S_{AB}^{T_B}\le\rho_{AB}^{T_B}\le S_{AB}^{T_B}$ and $  S_{AB}\ge 0$ are satisfied. Then $\operatorname{Tr}S_{AB}=1$, so
that $E_{\kappa}(\rho_{AB})\leq0$. Combining with the fact that $E_{\kappa
}(\rho_{AB})\geq0$ for all states, we conclude that $E_{\kappa}(\rho_{AB})=0$
if $\rho_{AB}^{T_{B}}\geq0$.

Finally, suppose that $E_{\kappa}(\rho_{AB})=0$. Then, by Proposition~\ref{prop:connect-to-log-neg}, $E_{N}(\rho_{AB})=0$, so
that $\left\Vert \rho_{AB}^{T_{B}}\right\Vert _{1}=1$. Decomposing $\rho
_{AB}^{T_{B}}$ into positive and negative parts as $\rho_{AB}^{T_{B}}=P-N$
(such that $P,N\geq0$ and $PN=0$), we have that $1=\operatorname{Tr}\rho
_{AB}=\operatorname{Tr}\rho_{AB}^{T_{B}}=\operatorname{Tr}P-\operatorname{Tr}%
N$. But we also have by assumption that $1=\left\Vert \rho_{AB}^{T_{B}%
}\right\Vert _{1}=\operatorname{Tr}P+\operatorname{Tr}N$. Subtracting these
equations gives $\operatorname{Tr}N=0$, which implies that $N=0$. From this,
we conclude that $\rho_{AB}^{T_{B}}=P\geq0$.
\end{proof}

\begin{proposition}[No convexity]
The $\kappa$-entanglement measure is not generally convex.
\end{proposition}
\begin{proof}
Due to Proposition~\ref{prop:connect-to-log-neg} and the fact that $|\rho_{AB}^{T_B}|^{T_B} \geq 0$ holds for any two-qubit state \cite{Ishizaka2004a},  the non-convexity of $\EK$ boils down to  finding a two-qubit example for which the logarithmic negativity is not convex. In particular, let us choose the two-qubit states $\rho_1=\Phi_2$,  $\rho_2=\frac{1}{2}(\proj{00}+\proj{11})$, and their average $\rho=\frac{1}{2} (\rho_1+\rho_2)$.
By direct calculation, we obtain 
\begin{align}
\EK(\rho_1)&=E_N(\rho_1)=1,\\
\EK(\rho_2)&=E_N(\rho_2)=0,\\
\EK(\rho)&=E_N(\rho)=\log_2 \frac{3}{2}.
\end{align}
Therefore, we have	
\begin{align}
\EK(\rho) 
		> \frac 1 2(\EK(\rho_1)+\EK(\rho_2)),
\end{align}
which concludes the proof.
\end{proof}

\bigskip

 An entanglement measure $E$ is monogamous \cite{CKW00,T04,KWin04} if the following inequality holds for all tripartite states $\rho_{ABC}$:
\begin{align}\label{eq:monogamy}
E(\rho_{AB})+E(\rho_{AC})\le E(\rho_{A(BC)}),
\end{align}
where the entanglement in $E(\rho_{A(BC)})$ is understood to be with respect to the bipartite cut between systems $A$ and $BC$.
It is known that some entanglement measures satisfy the monogamy inequality above \cite{CKW00,KWin04}.
However, the $\kappa$-entanglement measure is not monogamous, as we show by example in what follows.

\begin{proposition}[No monogamy]
The $\kappa$-entanglement measure is not generally monogamous.
\end{proposition}

\begin{proof}
Consider a state ${\ket\psi}{\bra \psi}_{ABC}$ of three qubits, where
\begin{align}
{\ket\psi}_{ABC}=\frac{1}{2}(\ket{000}_{ABC}+\ket{011}_{ABC}+\sqrt 2 \ket {110}_{ABC}).
\end{align}
Due to the fact that $\ket\psi_{ABC}$ can be written as
\begin{equation}
\ket\psi_{ABC} = [\ket 0_A \otimes \ket \Phi_{BC} + \ket 1_A \otimes \ket{10}_{BC}]/\sqrt{2},
\end{equation}
where $\ket{ \Phi}_{BC} = [\ket{00}_{BC} + \ket{11}_{BC} ] / \sqrt{2}$ and $\ket{ \Phi}_{BC}$ is orthogonal to $\ket{10}_{BC}$, this state is locally equivalent to $\ket{\Phi}_{AB}\otimes \ket 0_C$ with respect to the bipartite cut $A|BC$.
By direct calculation, one then finds that
\begin{align}
\EK(\psi_{A(BC)}) & =\EK(\Phi_{AB})=E_N(\Phi_{AB})=1 ,\\ 
	\EK(\psi_{AB}) & =E_N(\psi_{AB}) =\log_2 \frac 3 2, \\
	\EK(\psi_{AC}) & =E_N(\psi_{AC})= \log_2 \frac 3 2,
	\end{align}
which implies that
\begin{align}
\EK(\psi_{AB})+\EK(\psi_{AC}) >\EK(\psi_{A(BC)}).
\end{align}
This concludes the proof.
\end{proof}

\section{$\kappa$-entanglement measure is equal to the exact PPT-entanglement cost}

\begin{theorem}
\label{th:exact cost}
Let $\rho_{AB}$ be a  bipartite state acting on a separable Hilbert space. Then the exact PPT-entanglement cost of $\rho_{AB}$ is given by
	\begin{align}
	\label{eq:main result EPPT}		\EPPT(\rho_{AB})=E_\kappa(\rho_{AB}).
	\end{align}
\end{theorem}

\begin{proof}
The main idea behind the proof is to employ the one-shot bound in Proposition~\ref{lemma: one-shot} and then the additivity relation from Proposition~\ref{lemma:add}.
Consider that
\begin{align}
	\EPPT(\rho_{AB})&=\limsup_{n\to \infty}\frac{1}{n} \EPPTone(\rho_{AB}^{\ox n})\\
	&\le \limsup_{n\to \infty}\frac{1}{n}\log (2^{E_\kappa(\rho_{AB}^{\ox n})}+2)\\
	&=\limsup_{n\to \infty}\frac{1}{n}\log (2^{nE_\kappa(\rho_{AB})}+2)\\
	&=E_\kappa(\rho_{AB}).
	\end{align}
By a similar method, it is easy to show that $\EPPT(\rho_{AB})\ge E_\kappa(\rho_{AB})$.
\end{proof}

\section{Irreversibility of exact PPT entanglement manipulation}

\label{sec:examples-states irreversible}

The following example indicates the irreversibility of exact PPT entanglement manipulation, and it also implies that $\EPPT = E_{\kappa}$ is generally not equal to the logarithmic negativity $E_N$.
Consider the following rank-two state supported on the $3\times 3$ antisymmetric subspace \cite{Wang2016d}:
\begin{equation}
\rho^{v}_{AB} : =\frac{1}{2}(\proj{v_1}_{AB}+\proj{v_2}_{AB})
\end{equation}
with 
\begin{align}
\ket {v_1}_{AB} & := (\ket {01}_{AB}-\ket{10}_{AB})/{\sqrt 2},\\
 \ket {v_2}_{AB} & := (\ket {02}_{AB}-\ket{20}_{AB})/{\sqrt 2}.
\end{align}
For the state $\rho^{v}_{AB}$, 
the following holds
\begin{align}
R_{\max}(\rho^v_{AB}) & = E_N(\rho^{v}_{AB})
=\log_2\! \left(1+\frac{1}{\sqrt 2}\right) \\
& < \EPPT(\rho^v_{AB}) = E_{\kappa}(\rho^v_{AB}) = 1 \\
& < \log_2 Z(\rho^v_{AB})= 
\log_2 \!\left(1+\frac{13}{4\sqrt 2}\right),
\label{eq:strict-APE}
\end{align}
where $R_{\max}(\rho_v)$ denotes the max-Rains relative entropy \cite{WD16pra}.
The strict inequalities in \eqref{eq:strict-APE} also imply that both the lower and upper bounds from \eqref{eq:ape-bnds}, i.e., from \cite{Audenaert2003}, are generally not tight.

\section{Separations between $E_\kappa$ and $E_N$}

In this section, we show that  $E_\kappa$ is generally not equal to $E_N$. To do so, we employ numerical analysis of several classes of concrete examples. The codes for these numerical experiments are available \href{https://github.com/xinwang1/kappa-entanglement}{\textcolor{blue}{online}} and the SDPs are computed using the CVX software \cite{Grant2008}.

\textbf{Example 1}: Let us consider the following class of rank-two two-qutrit states:
\begin{align}
\sigma_p =p\proj{v_1}+(1-p)\proj{v_2},
\end{align}
with $\ket {v_1}={1}/{\sqrt 2}(\ket {01}-\ket{10}), \ket {v_2}={1}/{\sqrt 2}(\ket {02}-\ket{20})$. We show that $E_{\PPT}(\sigma_p) = E_\kappa(\sigma_p) > E_N(\sigma_p)$ for $0<p<1$ by the numerical comparison in Figure~\ref{fig:example 1}.

\begin{figure}
\centering
\includegraphics[width=0.45\textwidth]{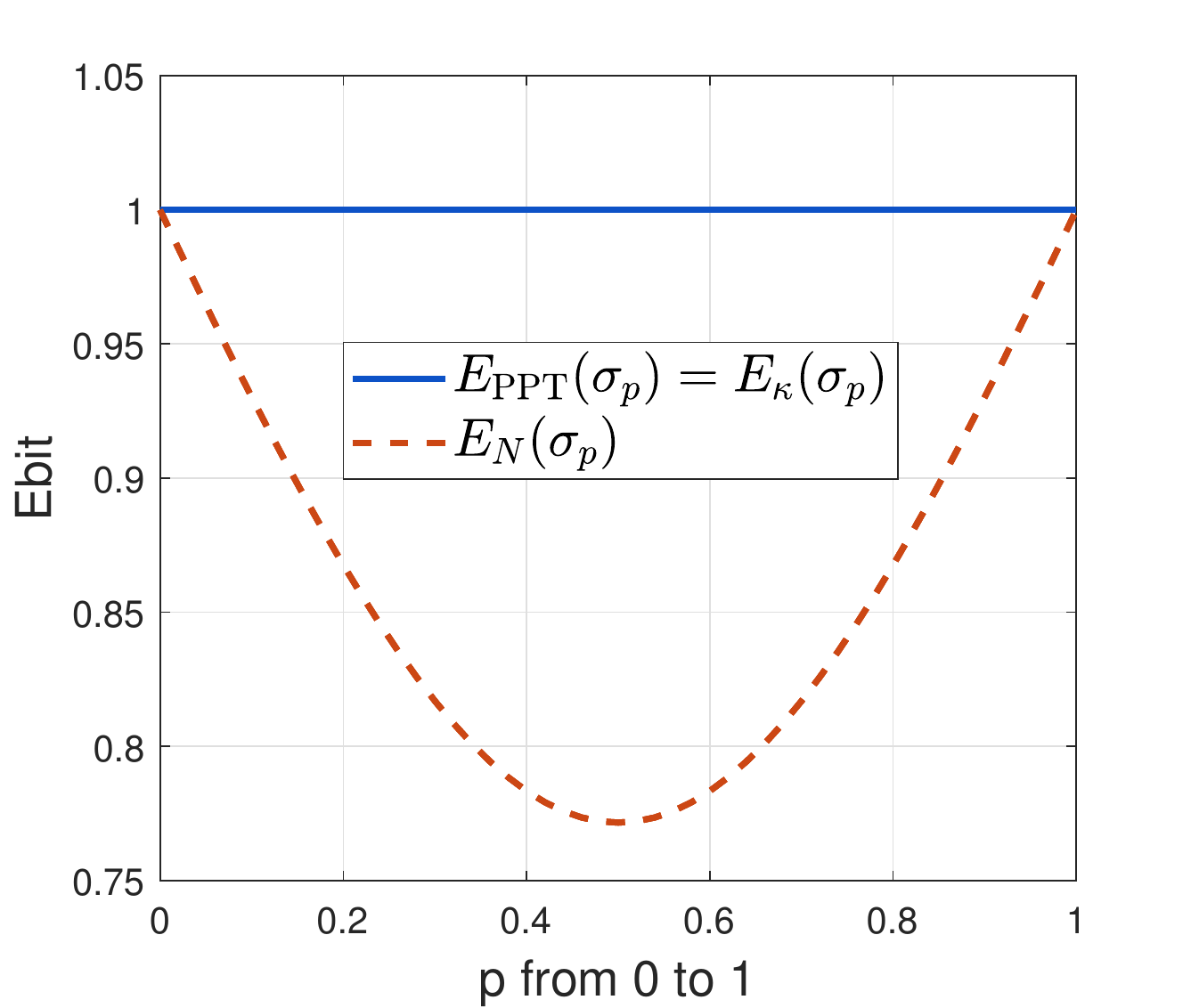}
\caption{$\kappa$-entanglement and  logarithmic negativity of $\sigma_p$}
\label{fig:example 1}
\end{figure}
	
\textbf{Example 2}: Let us consider the following class of rank-three two-qutrit states:
\begin{align}
\omega_p =\frac{p}{2}\proj{u_1}+\frac{1-p}{2}\proj{u_2} + \frac{1}{2}\proj{u_3},
\end{align}
with $\ket {u_1}=\ket{000}$, $\ket {u_2}={1}/{\sqrt 2}(\ket {02}+\ket{20})$, and $\ket{u_3}= {1}/{\sqrt 2}(\ket {12}+\ket{21})$.
The numerical comparison is presented in Figure~\ref{fig:example 2}.

\begin{figure}
\centering
\includegraphics[width=0.45\textwidth]{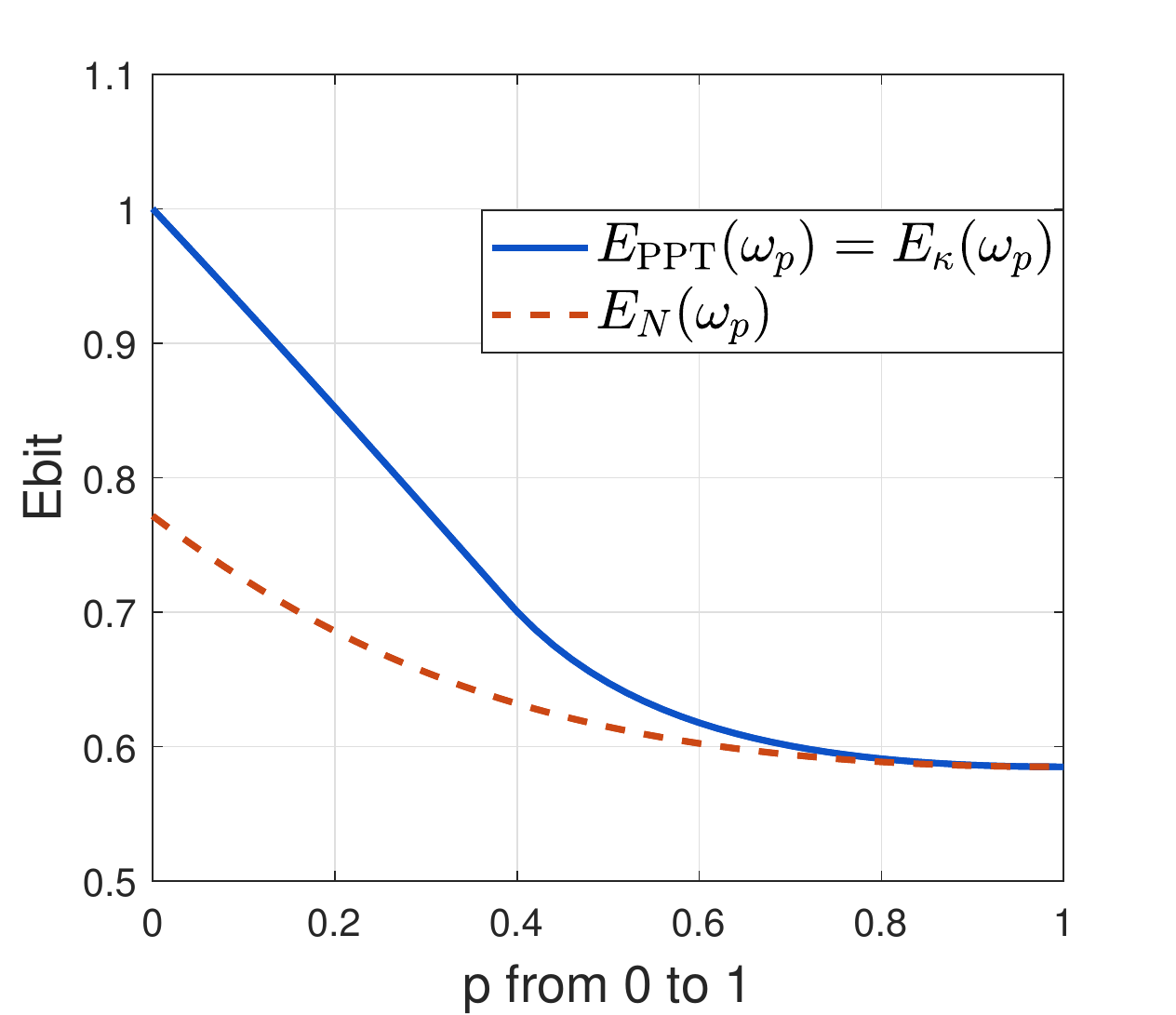}
\caption{$\kappa$-entanglement and  logarithmic negativity of $\omega_p$}
\label{fig:example 2}
\end{figure}
	
\textbf{Example 3}: Let us consider the following class of full-rank two-qutrit states:
\begin{align}
\tau_p =\frac{3p}{4}\proj{w_1}+\frac{3(1-p)}{4}\proj{w_2} + \frac{I/9}{4},
\end{align}
with $\ket {w_1}={1}/{\sqrt 3}(\ket {00}+\ket{12}+\ket{21})$, $\ket {w_2}={1}/{\sqrt 2}(\ket {02}+\ket{20})$.
The numerical comparison is presented  in Figure~\ref{fig:example 3}.

\begin{figure}
\centering
\includegraphics[width=0.45\textwidth]{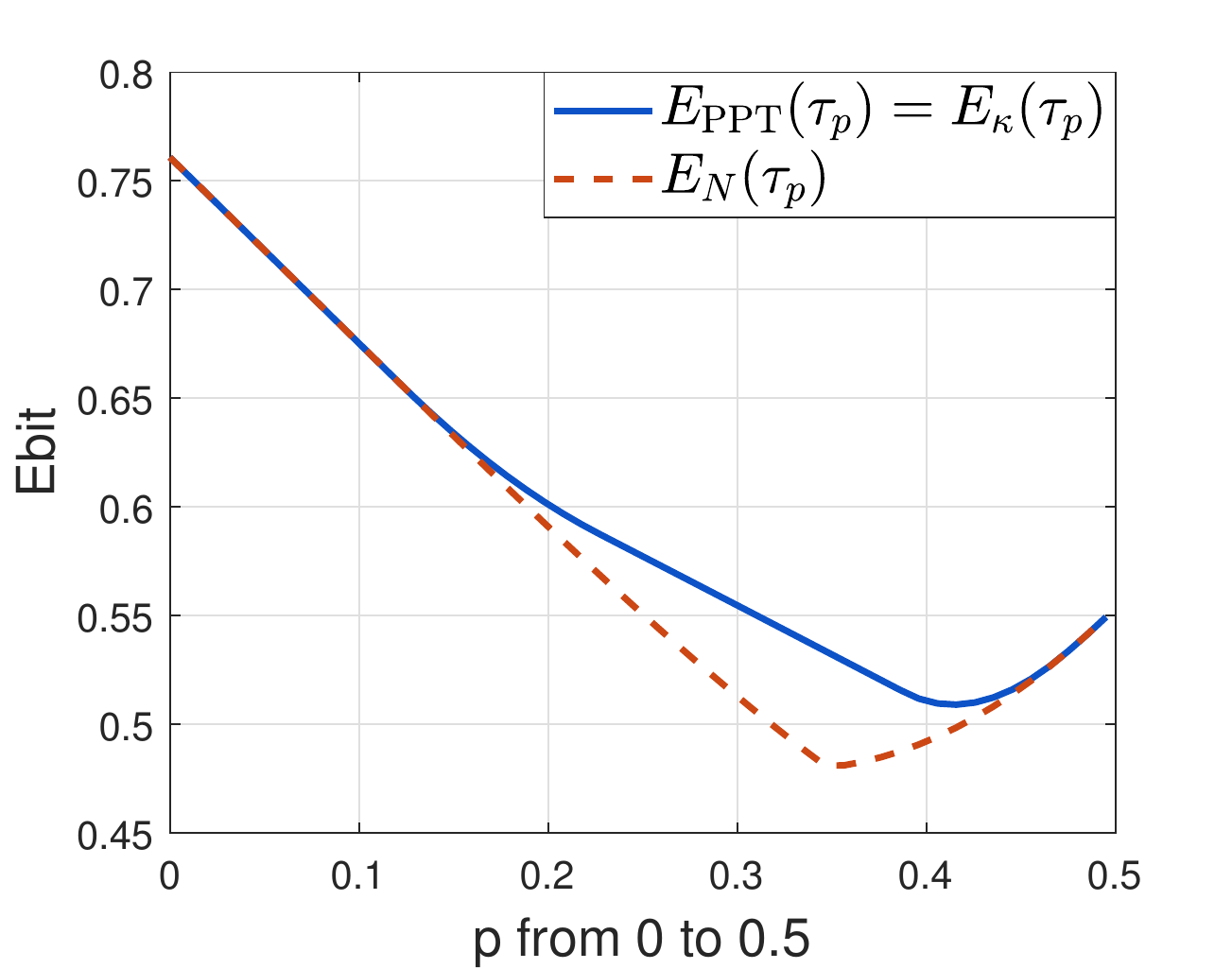}
\caption{$\kappa$-entanglement and  logarithmic negativity of $\tau_p$}
\label{fig:example 3}
\end{figure}

\section{Equality of $E_{\kappa}$ and $E^{\text{dual}}_{\kappa}$ for states acting on separable Hilbert spaces}

\label{app:kappa-to-dual-infty}

In this section, we prove that%
\begin{equation}
E_{\kappa}(\rho_{AB})=E_{\kappa}^{\text{dual}}(\rho_{AB}%
),\label{eq:kappa-dual-infty}%
\end{equation}
for a state $\rho_{AB}$ acting on a separable Hilbert space. To begin with,
let us recall that the following inequality always holds from weak duality%
\begin{equation}
E_{\kappa}(\rho_{AB})\geq E_{\kappa}^{\text{dual}}(\rho_{AB}).
\end{equation}
So our goal is to prove the opposite inequality. We suppose throughout that
$E_{\kappa}^{\text{dual}}(\rho_{AB})<\infty$. Otherwise, the desired equality
in \eqref{eq:kappa-dual-infty}\ is trivially true. We also suppose that
$\rho_{AB}$ has full support. Otherwise, it is finite-dimensional and the
desired equality in \eqref{eq:kappa-dual-infty}\ is trivially true, or it has only finitely many zero entries, in which case it is isomorphic to a state with full support.

To this end, consider sequences $\{\Pi_{A}^{k}\}_{k}$ and $\{\Pi_{B}^{k}%
\}_{k}$ of projectors weakly converging to the identities $\1_{A}$ and $\1_{B}$
and such that $\Pi_{A}^{k}\leq\Pi_{A}^{k^{\prime}}$ and $\Pi_{B}^{k}\leq
\Pi_{B}^{k^{\prime}}$ for $k^{\prime}\geq k$. Furthermore, we suppose that
$[\Pi_{B}^{k}]^{T_{B}}=\Pi_{B}^{k}$ for all $k$. Then define%
\begin{equation}
\rho_{AB}^{k}\coloneqq\left(  \Pi_{A}^{k}\otimes\Pi_{B}^{k}\right)  \rho_{AB}\left(
\Pi_{A}^{k}\otimes\Pi_{B}^{k}\right)  .
\end{equation}
It follows that \cite{D67}%
\begin{equation}
\lim_{k\rightarrow\infty}\left\Vert \rho_{AB}-\rho_{AB}^{k}\right\Vert _{1}=0.
\end{equation}

We now prove that%
\begin{equation}
E_{\kappa}^{\text{dual}}(\rho_{AB})\geq E_{\kappa}^{\text{dual}}(\rho_{AB}%
^{k})\label{eq:dual-projection}%
\end{equation}
for all $k$. Let $A^{k}$ and $B^{k}$ denote the subspaces onto which $\Pi
_{A}^{k}$ and $\Pi_{B}^{k}$ project. Let $V_{A^{k}B^{k}}^{k}$ and
$W_{A^{k}B^{k}}^{k}$ be arbitrary operators satisfying $V_{AB}^{k}+W_{AB}%
^{k}\leq \1_{A^{k}B^{k}}=\left(  \Pi_{A}^{k}\otimes\Pi_{B}^{k}\right)  $,
$[V_{A^{k}B^{k}}^{k}]^{T_{B}},[W_{A^{k}B^{k}}^{k}]^{T_{B}}\geq0$. Set%
\begin{align}
\overline{V}_{AB}^{k}  & \coloneqq\left(  \Pi_{A}^{k}\otimes\Pi_{B}^{k}\right)
V_{A^{k}B^{k}}^{k}\left(  \Pi_{A}^{k}\otimes\Pi_{B}^{k}\right)  ,\\
\overline{W}_{AB}^{k}  & \coloneqq\left(  \Pi_{A}^{k}\otimes\Pi_{B}^{k}\right)
W_{A^{k}B^{k}}^{k}\left(  \Pi_{A}^{k}\otimes\Pi_{B}^{k}\right)  ,
\end{align}
and note that%
\begin{align}
\overline{V}_{AB}^{k}+\overline{W}_{AB}^{k}  & \leq \1_{AB},\\
\lbrack\overline{V}_{AB}^{k}]^{T_{B}},[\overline{W}_{AB}^{k}]^{T_{B}}  &
\geq0.
\end{align}
Then%
\begin{align}
\operatorname{Tr}\rho_{AB}^{k}(V_{A^{k}B^{k}}^{k}-W_{A^{k}B^{k}}^{k})  &
=\operatorname{Tr}\left(  \Pi_{A}^{k}\otimes\Pi_{B}^{k}\right)  \rho
_{AB}\left(  \Pi_{A}^{k}\otimes\Pi_{B}^{k}\right)  (V_{A^{k}B^{k}}%
^{k}-W_{A^{k}B^{k}}^{k})\\
& =\operatorname{Tr}\rho_{AB}\left(  \Pi_{A}^{k}\otimes\Pi_{B}^{k}\right)
(V_{A^{k}B^{k}}^{k}-W_{A^{k}B^{k}}^{k})\left(  \Pi_{A}^{k}\otimes\Pi_{B}%
^{k}\right)  \\
& =\operatorname{Tr}\rho_{AB}(\overline{V}_{AB}^{k}-\overline{W}_{AB}^{k})\\
& \leq E_{\kappa}^{\text{dual}}(\rho_{AB}).
\end{align}
Since the inequality holds for arbitrary $V_{A^{k}B^{k}}^{k}$ and
$W_{A^{k}B^{k}}^{k}$ satisfying the conditions above, we conclude the
inequality in \eqref{eq:dual-projection}.

Thus, we conclude that%
\begin{equation}
E_{\kappa}^{\text{dual}}(\rho_{AB})\geq\limsup_{k\rightarrow\infty}E_{\kappa
}^{\text{dual}}(\rho_{AB}^{k}).\label{eq:lim-sup-lower}%
\end{equation}

Now let us suppose that $E_{\kappa}^{\text{dual}}(\rho_{AB})<\infty$. Then for
all $V_{AB}$ and $W_{AB}$ satisfying $V_{AB}+W_{AB}\leq \1_{AB}$,
$[V_{AB}]^{T_{B}},[W_{AB}]^{T_{B}}\geq0$, as well as $\operatorname{Tr}%
\rho_{AB}(V_{AB}-W_{AB})\geq0$, we have that%
\begin{equation}
\operatorname{Tr}\rho_{AB}(V_{AB}-W_{AB})<\infty.
\end{equation}
Since $\rho_{AB}$ has full support, this means that%
\begin{equation}
\left\Vert V_{AB}-W_{AB}\right\Vert _{\infty}<\infty.
\end{equation}
Considering that from H\"older's inequality%
\begin{equation}
\left\vert \operatorname{Tr}(\rho_{AB}-\rho_{AB}^{k})(V_{AB}-W_{AB}%
)\right\vert \leq\left\Vert \rho_{AB}-\rho_{AB}^{k}\right\Vert _{1}\left\Vert
V_{AB}-W_{AB}\right\Vert _{\infty},
\end{equation}
and setting%
\begin{align}
V_{AB}^{k}  & \coloneqq\left(  \Pi_{A}^{k}\otimes\Pi_{B}^{k}\right)  V_{AB}\left(
\Pi_{A}^{k}\otimes\Pi_{B}^{k}\right)  ,\\
W_{AB}^{k}  & \coloneqq\left(  \Pi_{A}^{k}\otimes\Pi_{B}^{k}\right)  W_{AB}\left(
\Pi_{A}^{k}\otimes\Pi_{B}^{k}\right)  ,
\end{align}
we conclude that%
\begin{align}
\operatorname{Tr}\rho_{AB}(V_{AB}-W_{AB})  & \leq\liminf_{k\rightarrow\infty
}\operatorname{Tr}\rho_{AB}^{k}(V_{AB}-W_{AB})\\
& =\liminf_{k\rightarrow\infty}\operatorname{Tr}\rho_{AB}^{k}(V_{AB}%
^{k}-W_{AB}^{k})\\
& \leq\liminf_{k\rightarrow\infty}\sup_{V^{k},W^{k}}\operatorname{Tr}\rho
_{AB}^{k}(V_{AB}^{k}-W_{AB}^{k})\\
& =\liminf_{k\rightarrow\infty}E_{\kappa}^{\text{dual}}(\rho_{AB}^{k}).
\end{align}
Since the inequality holds for arbitrary $V_{AB}$ and $W_{AB}$ satisfying the
above conditions, we conclude that%
\begin{equation}
E_{\kappa}^{\text{dual}}(\rho_{AB})\leq\liminf_{k\rightarrow\infty}E_{\kappa
}^{\text{dual}}(\rho_{AB}^{k}).\label{eq:lim-inf-upper}%
\end{equation}

Putting together \eqref{eq:lim-sup-lower} and \eqref{eq:lim-inf-upper}, we
conclude that%
\begin{equation}
E_{\kappa}^{\text{dual}}(\rho_{AB})=\lim_{k\rightarrow\infty}E_{\kappa
}^{\text{dual}}(\rho_{AB}^{k}).\label{eq:dual-limits}%
\end{equation}

From strong duality for the finite-dimensional case, we have for all $k$ that%
\begin{equation}
E_{\kappa}^{\text{dual}}(\rho_{AB}^{k})=E_{\kappa}(\rho_{AB}^{k}),
\end{equation}
and thus that%
\begin{equation}
\lim_{k\rightarrow\infty}E_{\kappa}^{\text{dual}}(\rho_{AB}^{k})=\lim
_{k\rightarrow\infty}E_{\kappa}(\rho_{AB}^{k}).\label{eq:finite-str-dual}%
\end{equation}

It thus remains to prove that%
\begin{equation}
\lim_{k\rightarrow\infty}E_{\kappa}(\rho_{AB}^{k})=E_{\kappa}(\rho_{AB}).
\end{equation}
We first prove that%
\begin{equation}
E_{\kappa}(\rho_{AB})\geq\limsup_{k\rightarrow\infty}E_{\kappa}(\rho_{AB}%
^{k}).\label{eq:lim-sup-lower_e_kappa}%
\end{equation}
Let $S_{AB}$ be an arbitrary operator satisfying
\begin{equation}
S_{AB}\geq0,\qquad-S_{AB}^{T_{B}}\leq\rho_{AB}^{T_{B}}\leq S_{AB}^{T_{B}}.
\label{eq:conditions-for-S-limits}
\end{equation}
Then, defining $S_{AB}^{k}=\left(  \Pi_{A}^{k}\otimes\Pi_{B}^{k}\right)
S_{AB}\left(  \Pi_{A}^{k}\otimes\Pi_{B}^{k}\right)  $, we have that%
\begin{equation}
S_{AB}^{k}\geq0,\qquad-[S_{AB}^{k}]^{T_{B}}\leq\lbrack\rho_{AB}^{k}]^{T_{B}%
}\leq\lbrack S_{AB}^{k}]^{T_{B}}.
\end{equation}
Then%
\begin{equation}
\log_2\operatorname{Tr}S_{AB}\geq\log_2\operatorname{Tr}S_{AB}^{k}\geq E_{\kappa
}(\rho_{AB}^{k}).
\end{equation}
Since the inequality holds for all $S_{AB}$ satisfying \eqref{eq:conditions-for-S-limits}, we conclude that%
\begin{equation}
E_{\kappa}(\rho_{AB})\geq E_{\kappa}(\rho_{AB}^{k})
\end{equation}
for all $k$, and thus \eqref{eq:lim-sup-lower_e_kappa}\ holds.

The rest of the proof follows \cite{FAR11} closely.
Since the
condition $\Pi_{A}^{k}\leq\Pi_{A}^{k^{\prime}}$ and $\Pi_{B}^{k}\leq\Pi
_{B}^{k^{\prime}}$ for $k^{\prime}\geq k$ holds, in fact the same sequence of
steps as above allows for concluding that%
\begin{equation}
E_{\kappa}(\rho_{AB}^{k^{\prime}})\geq E_{\kappa}(\rho_{AB}^{k}),
\end{equation}
meaning that the sequence is monotone non-decreasing with $k$. Thus, we can
define%
\begin{equation}
\mu\coloneqq\lim_{k\rightarrow\infty}E_{\kappa}(\rho_{AB}^{k})\in\mathbb{R}^{+},
\end{equation}
and note from the above that%
\begin{equation}
\mu\leq E_{\kappa}(\rho_{AB}).
\end{equation}
For each $k$, let $S_{AB}^{k}$ denote an optimal operator such that
$E_{\kappa}(\rho_{AB}^{k})=\log_2\operatorname{Tr}S_{AB}^{k}$. From the fact
that $S_{AB}^{k}\geq0$, and $\operatorname{Tr}S_{AB}^{k}\leq2^{\mu}$, we
conclude that $\{S_{AB}^{k}\}_{k}$ is a bounded sequence in the trace class
operators. Since the trace class operators form the dual space of the compact
operators $\mathcal{K}(\mathcal{H}_{AB})$ \cite{RS78}, we can apply the Banach--Alaoglu
theorem \cite{RS78} to find a subsequence $\{S_{AB}^{k}\}_{k\in\Gamma}$ with a
weak$^{\ast}$ limit $\widetilde{S}_{AB}$ in the trace class operators such that
$\widetilde{S}_{AB}\geq0$ and $\operatorname{Tr}[\widetilde{S}_{AB}]\leq2^{\mu}$.
Furthermore, the sequences $[\rho_{AB}^{k}]^{T_{B}}+[S_{AB}^{k}]^{T_{B}}$ and
$[S_{AB}^{k}]^{T_{B}}-[\rho_{AB}^{k}]^{T_{B}}$ converge in the weak operator
topology to $\rho_{AB}^{T_{B}}+\widetilde{S}_{AB}^{ T_{B}}$ and $\widetilde{S}_{AB}^{ T_{B}%
}-\rho_{AB}^{T_{B}}$, respectively, and we can then conclude that $\rho
_{AB}^{T_{B}}+\widetilde{S}_{AB}^{ T_{B}},\widetilde{S}_{AB}^{ T_{B}}-\rho_{AB}^{T_{B}}\geq0$.
But this means that%
\begin{equation}
E_{\kappa}(\rho_{AB})\leq\log_2\operatorname{Tr}\widetilde{S}_{AB}\leq\mu,
\end{equation}
which implies that%
\begin{equation}
E_{\kappa}(\rho_{AB})\leq\liminf_{k\rightarrow\infty}E_{\kappa}(\rho_{AB}%
^{k}).\label{eq:lim-inf-upper_e_kappa}%
\end{equation}
Putting together
\eqref{eq:lim-sup-lower_e_kappa} and \eqref{eq:lim-inf-upper_e_kappa}, we
conclude that%
\begin{equation}
E_{\kappa}(\rho_{AB})=\lim_{k\rightarrow\infty}E_{\kappa}(\rho_{AB}%
^{k}).\label{eq:primal-limits}%
\end{equation}
Finally, putting together \eqref{eq:dual-limits}, \eqref{eq:finite-str-dual},
and \eqref{eq:primal-limits},\ we conclude \eqref{eq:kappa-dual-infty}.


\end{document}